\documentclass[runningheads]{llncs}

\newif\ifArxiv
\Arxivtrue

\usepackage{amssymb,amsfonts,amsmath}
\usepackage{booktabs,tabularx}
\usepackage{xspace,soul,wrapfig}
\usepackage{graphicx,paralist,enumerate}
\usepackage[hidelinks]{hyperref}
\usepackage[textsize=footnotesize]{todonotes}
\usepackage[full]{complexity}
\usepackage[capitalise]{cleveref}
\newcommand{\bigoh}{\mathcal{O}}
\usepackage[utf8]{inputenc}
\usepackage{subcaption}
\usepackage{cite}
\usepackage{thmtools,thm-restate}
\usepackage{url,soul}

\graphicspath{{./Figures/}}

\newcommand{\midset}{\ensuremath{\operatorname{mid}}}

\newcommand{\sig}{\ensuremath{\operatorname{sig}}}

\newcommand{\outerp}{\ensuremath{\Gamma_{\operatorname{op}}}}
\newcommand{\bless}{\ensuremath{\Gamma_{\operatorname{bl}}}}

\newcommand{\planarreemb}{\ensuremath{\Gamma_{\operatorname{re}}}}

\newcommand{\candidate}{\ensuremath{S}}
\newcommand{\splt}{\ensuremath{\curlywedge}}
\newcommand{\ssplit}{\ensuremath{{S_{\splt}}}}
\newcommand{\ssplitp}{\ensuremath{{S'_{\splt}}}}
\newcommand{\ssplitsol}{\ensuremath{{S_{\splt}^*}}}

\newcommand{\addit}{\ensuremath{S_{\operatorname{add}}}}
\newcommand{\cycle}{\ensuremath{C_{\operatorname{out}}}}
\newcommand{\inner}{\ensuremath{S_{\operatorname{in}}}}
\newcommand{\souter}{\ensuremath{S_{\operatorname{out}}}}
\newcommand{\noose}{\ensuremath{N_\eta}}

\newcommand{\sn}{\textsc{{Splitting Number}}}

\newcommand{\copies}{\ensuremath{\textsf{copies}}}
\newcommand{\orig}{\ensuremath{\textsf{orig}}}
\newcommand{\copiessol}{\ensuremath{\textsf{copies}}^*}
\newcommand{\origsol}{\ensuremath{\textsf{orig}}^*}

\newcommand{\algA}{\ensuremath{\mathcal{A}}}
\newcommand{\algB}{\ensuremath{\mathcal{B}}}

\newcommand{\ps}{\ensuremath{p_{\text{s}}}}
\newcommand{\pe}{\ensuremath{p_{\text{e}}}}

\newcommand\pESN{\textsc{Embedded Splitting Number}}

\ifArxiv
\renewcommand{\orcidID}[1]{\href{https://orcid.org/#1}{\includegraphics[scale=.03]{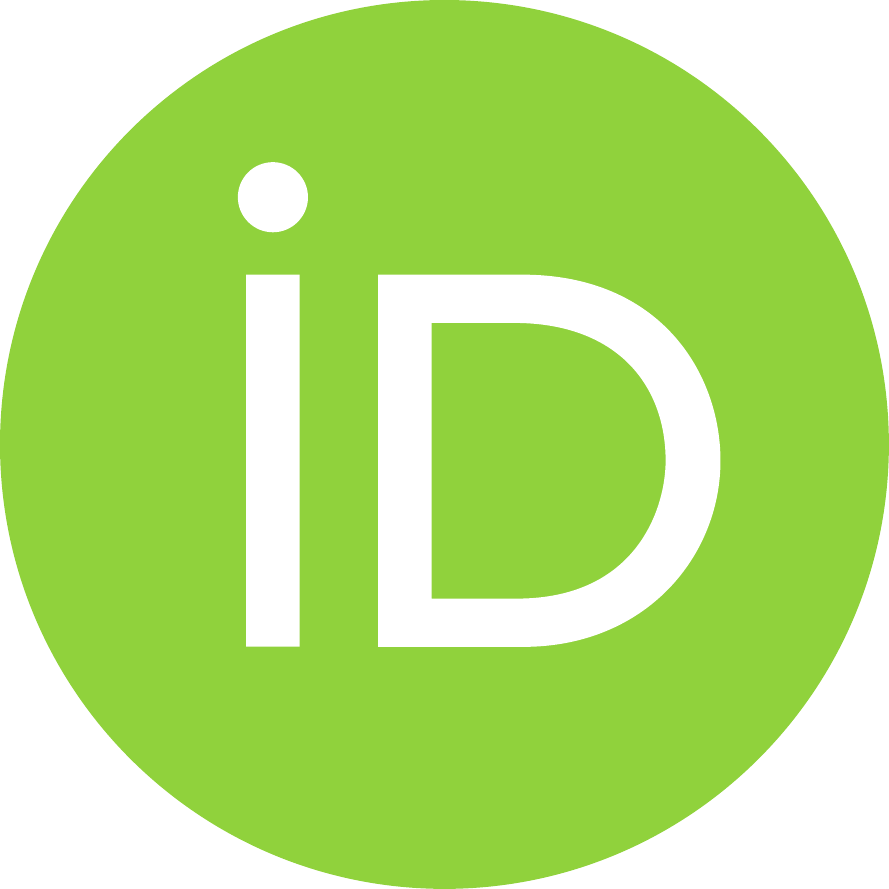}}}
\fi

\spnewtheorem{branchingR}{Branching Rule}{\bfseries}{\itshape}
\spnewtheorem{reduction}{Reduction Rule}{\bfseries}{\itshape}

\usepackage{etoolbox} %
\newcommand{\toappendix}[1]{%
  \gappto{\appendixProofText}{
    {#1}
   }
}

\newcommand{\appendixproofwithrestatable}[2]{%
  \gappto{\appendixProofText}{
    {#1}
    {#2}
   }
}

\newcommand{\appendixsection}[1]{%
  \gappto{\appendixProofText}{
    \section{Additional Material for Section~\ref{#1}}
    \label{appsec:#1}
  }
}

\begin{document}

\title{Planarizing Graphs and their Drawings by Vertex Splitting \thanks{Manuel Sorge is supported by the Alexander von Humboldt Foundation, Anaïs Villedieu is supported by the Austrian Science Fund (FWF) under grant P31119, Jules Wulms is partially supported by the Austrian Science Fund (FWF) under grant P31119 and partially by the Vienna Science and Technology Fund (WWTF) under grant ICT19-035. Colored versions of figures can be found in the online version of this paper.}}

\author{Martin Nöllenburg\inst{1}\orcidID{0000-0003-0454-3937} \and
Manuel Sorge\inst{1}\orcidID{0000-0001-7394-3147} \and
Soeren Terziadis\inst{1}\orcidID{0000-0001-5161-3841} \and 
Anaïs~Villedieu\inst{1}\orcidID{0000-0001-6196-8347} \and
Hsiang-Yun Wu\inst{2,3}\orcidID{0000-0003-1028-0010} \and
Jules Wulms\inst{1}\orcidID{0000-0002-9314-8260}}
\authorrunning{M. N\"ollenburg et al.}

\institute{Algorithms and Complexity Group, TU Wien, Vienna, Austria
\email{\{noellenburg|manuel.sorge|sterziadis|avilledieu|jwulms\}@ac.tuwien.ac.at}\and
Research Unit of Computer Graphics, TU Wien, Vienna, Austria \and
St.\ P{\"o}lten University of Applied Sciences, Austria\\
\email{hsiang.yun.wu@acm.org}}

\maketitle              %
\begin{abstract}
The splitting number of a graph $G=(V,E)$ is the minimum number of vertex splits required to turn $G$ into a planar graph, where a vertex split removes a vertex $v \in V$, introduces two new vertices $v_1, v_2$, and distributes the edges formerly incident to $v$ among 
$v_1, v_2$.
The splitting number problem is known to be \NP-complete for abstract graphs and we provide a non-uniform fixed-parameter tractable (\FPT) algorithm for this problem.
We then shift focus to the splitting number of a given topological graph drawing in $\mathbb{R}^2$,
where the new vertices resulting from vertex splits must be re-embedded into the existing drawing of the remaining graph. 
We show \NP-completeness of this \emph{embedded} %
splitting number problem, even for its two subproblems of (1) selecting a minimum subset of vertices to split and (2) for re-embedding a minimum number of copies of a given set of vertices. %
For the latter problem we present an \FPT\ algorithm parameterized by the number of vertex splits.
This algorithm reduces to a bounded outerplanarity case and uses an intricate dynamic program on a sphere-cut decomposition.

\end{abstract}

\section{Introduction}
\looseness=-1

While planar graphs admit compact and naturally crossing-free drawings, computing good layouts of large and dense non-planar graphs remains a challenging task, mainly due to the visual clutter caused by large numbers of edge crossings. However, graphs in many applications are typically non-planar and hence several methods have been proposed to simplify their drawings and minimize crossings, both from a practical point of view~\cite{vonLandesberger:2011:CGF,lht-setbt-17} and a theoretical one~\cite{s-cng-18,Liebers_2001}. Drawing algorithms often focus on reducing the number of visible crossings~\cite{n-clng-20} or improving crossing angles~\cite{o-aravc-20}, aiming to achieve similar beneficial readability properties as in crossing-free drawings of planar graphs.

One way of turning a non-planar graph into a planar one while retaining the entire graph and not deleting any of its vertices or edges, is to apply a sequence of \emph{vertex splitting} operations, a technique which has been studied in theory~\cite{eades1995vertex,Eppstein_2017,Liebers_2001,ku-twcg-16}, but which is also used in practice, e.g., by biologists and social scientists~\cite{rd-ued-10,HenryBF08,Nielsen_2019_CBB,Wu_2019_BMC,Wu_2020_TVCG}.
For a given graph $G=(V,E)$ and a vertex $v \in V$, a \emph{vertex split} of $v$ removes $v$ from $G$ and instead adds two non-adjacent copies $v_1, v_2$ such that the edges formerly incident to $v$ are distributed among $v_1$ and $v_2$.
Similarly, a \emph{$k$-split} of $v$ for $k \ge 2$ creates $k$ copies $v_1, \dots, v_k$, among which the edges formerly incident to $v$ are distributed.
On the one hand, splitting a vertex can resolve some of the crossings of its incident edges, but on the other hand the number of objects in the drawing to keep track of increases. Therefore, we aim to minimize the number of splits needed to obtain a planar graph, which is known as the \emph{splitting number} of the graph. Computing it is \NP-hard~\cite{Faria_2001}, but it is known for some graph classes including complete and complete bipartite graphs~\cite{HartsfieldJR85,jackson1984splitting,Hartsfield86}.
A related concept is the folded covering number~\cite{ku-twcg-16} or equivalently the planar split thickness~\cite{Eppstein_2017} of a graph $G$, which is the minimum $k$ such that $G$ can be decomposed into at most $k$ planar subgraphs by applying a $k$-split to each vertex of $G$ at most once. %
Eppstein et al.~\cite{Eppstein_2017} showed that deciding whether a graph has split thickness $k$ is \NP-complete, even for $k=2$, but can be approximated within a constant factor and is fixed-parameter tractable (\FPT) for graphs of bounded treewidth.

While previous work considered vertex splitting in the context of abstract graphs, our focus in this paper is on vertex splitting for non-planar, topological graph drawings in $\mathbb R^2$. 
In this case we want to improve the given input drawing by applying changes to a minimum number of split vertices, which can be freely re-embedded, while the non-split vertices must remain at their original positions in order to maintain layout stability%
~\cite{mels-lam-95}, see \cref{fig:basic}. %

\begin{figure}[tb]
    \centering
    \includegraphics{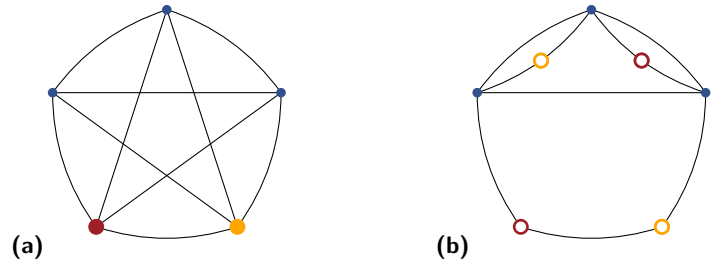}
    \caption{Vertex splitting in a drawing of $K_5$. The red and orange disks in \textbf{\textsf{(a)}} are split once into the red and orange circles in \textbf{\textsf{(b)}}. Note that an abstract $K_5$ without drawing has splitting number 1.
    }
    \label{fig:basic}
\end{figure}

\newcommand{\pVP}{\textsc{Vertex Planarization}}
\looseness=-1
The underlying algorithmic problem for vertex splitting in drawings of graphs is two-fold: firstly, a suitable (minimum) subset of vertices to be split must be selected, and secondly the split copies of these vertices must be re-embedded in a crossing-free way together with a partition of the original edges of each split vertex into a subset for each of its copies. 

The former problem is closely related to the \NP-complete problem \pVP, where we want to decide whether a given graph can be made planar by deleting at most $k$ vertices, and to related problems of hitting graph minors by vertex deletions.
Both are very well-studied in the parameterized complexity realm~\cite{robertson_graph_1995,kawarabayashi_planarity_2009,jansen_nearoptimal_2013a,kim_linear_2015}.
For example, it follows from results of Robertson and Seymour~\cite{robertson_graph_1995} that \pVP\ can be solved in cubic time for fixed $k$ and a series of papers~\cite{marx_obtaining_2012,kawarabayashi_planarity_2009,jansen_nearoptimal_2013a} improved the dependency on the input size to linear and the dependency on $k$ to $2^{O(k \log k)}$.

The latter re-embedding problem is related to drawing extension problems, where a subgraph is drawn and the missing vertices and edges must be inserted in a (near-)planar way into this drawing~\cite{adfjkp-tppeg-15,eghkn-ep1d-20,eghkn-enc1dpt-20,DBLP:conf/wg/ArroyoKPSVW20,Chimani_2009,Chimani_2015}. In these works, however, the incident edges of each vertex are given, while we still need to distribute them among the copies. Furthermore, as we show in~\cref{sec:complex}, it generalizes natural problems on covering vertices by faces in planar graphs~\cite{bm-ccvfpg-88,kloks_new_2002,abu-khzam_direct_2004,abu-khzam_bounded_2008}. %

\subsubsection{Contributions.}
\looseness=-1
In this paper we extend the investigation of the splitting number problem and its complexity from abstract graphs to graphs with a given (non-planar) topological drawing.
In~\cref{sec:nonunifpt}, we first show that the original splitting number problem is non-uniformly \FPT\ when parameterized by the number of split operations using known results on minor-closed graph classes.
We then describe a polynomial-time algorithm for minimizing crossings in a given drawing when re-embedding the copies of a single vertex, split at most $k$ times for a fixed integer $k$, in~\cref{sec:poly-alg}. 

For the remainder of the paper we shift our focus to two basic subproblems of vertex splitting in topological graph drawings. %
We distinguish the \emph{candidate selection step}, where we want to compute a set of vertices that requires the minimum number of splits to obtain planarity, and the \emph{re-embedding step}, that asks where each copy should be put back into the drawing and with which neighborhood. 
We prove  in~\cref{sec:complex} that both problems are \NP-complete, using a reduction from vertex cover in planar cubic graphs for the candidate selection problem and showing that the re-embedding problem generalizes the %
\NP-complete \textsc{Face Cover} problem.
Finally, in~\cref{sec:re-embed-fpt} we present an \FPT~algorithm for the re-embedding problem %
parameterized by the number of splits.
Given a partial planar drawing and a set of vertices to split and re-embed, the algorithm first reduces the instance to a bounded-outerplanarity case and then applies dynamic programming on the decomposition tree of a sphere-cut decomposition of the remaining partial drawing.
We note that our reduction for showing \NP-hardness of the re-embedding problem is indeed a parameterized reduction from \textsc{Face Cover} parameterized by %
the solution size to the re-embedding problem parameterized by the number of allowed splits.
\textsc{Face Cover} is known to be \FPT\ %
in this case~\cite{abu-khzam_bounded_2008}, and hence our \FPT\ algorithm is a generalization of that result.

\smallskip 
\ifArxiv \emph{Due to space constraints, missing proofs and details are found in the appendix.} \fi

\section{Preliminaries}\label{sec:prelims}

\looseness=-1
Let $G=(V,E)$ be a simple graph with vertex set $V(G) = V$ and edge set $E(G) = E$.
For a subset $V' \subset V$, $G[V']$ denotes the subgraph of $G$ induced by $V'$.
The neighborhood of a vertex $v\in V(G)$ is defined as $N_G(v)$.
If $G$ is clear from the context, we omit the subscript $G$.
A \emph{split} operation applied to a vertex $v$ results in a graph $G'=(V',E')$ where $V' = V \setminus \{v\} \cup \{\dot{v}^{(1)},\dot{v}^{(2)}\}$ and $E'$ is obtained from $E$ by distributing the edges incident to $v$ among $\dot{v}^{(1)},\dot{v}^{(2)}$ such that $N_G(v)=N_{G'}(\dot{v}^{(1)})\cup N_{G'}(\dot{v}^{(2)})$ (copies are written with a dot for clarity). Splits with $N(\dot{v}^{(1)})=N(v)$ and $N(\dot{v}^{(2)})=\emptyset$ (equivalent to moving $v$ to $\dot{v}^{(1)}$), or with $N(\dot{v}^{(1)})\cap N(\dot{v}^{(2)})\ne\emptyset$ (which is never beneficial, but can simplify proofs) are allowed.
The vertices $\dot{v}^{(1)}, \dot{v}^{(2)}$ are called \emph{split vertices} or \emph{copies} of $v$.
If a copy $\dot{v}$ of a vertex $v$ is split again, then any copy of $\dot{v}$ is also called a copy of the original vertex $v$ and we use the notation $\dot{v}^{(i)}$ for $i= 1,2, \dots$ to denote the different copies of $v$.

\begin{problem}[\textsc{Splitting Number}]
Given a graph $G=(V,E)$ and an integer $k$, can $G$ be transformed into a planar graph $G'$ by applying at most $k$ splits to $G$?
\end{problem}

\sn~is \NP-complete, even for cubic graphs~\cite{Faria_2001}. %
We extend the notion of vertex splitting to drawings of graphs. 
Let $G$ be a graph and let $\Gamma$ be a \emph{topological drawing} of $G$, 
which maps each vertex to a point in $\mathbb R^2$ and each edge to a simple curve connecting the points corresponding to the incident vertices of that edge.
We still refer to the points and curves as vertices and edges, respectively, in such a drawing.
Furthermore, we assume $\Gamma$ is a \emph{simple} drawing, meaning no two edges intersect more than once, no three edges intersect in one point (except common endpoints), and adjacent edges do not cross. 

\begin{problem}[\textsc{Embedded Splitting Number}]\label{pb:esn}
Given a graph $G=(V,E)$ with a simple topological drawing $\Gamma$ and an integer $k$, can $G$ be transformed into a graph %
$G'$ by applying at most $k$ splits to $G$ such that $G'$ has a planar drawing that coincides with $\Gamma$ on $G[V(G) \cap V(G')]$?
\end{problem}%
Problem~\ref{pb:esn} includes two interesting subproblems, namely an embedded vertex deletion problem (which corresponds to selecting  candidates for splitting) %
and a subsequent re-embedding problem, both defined below.

\newcommand\pCS{\textsc{Embedded Vertex Deletion}}
\begin{problem}[\pCS]\label{pb:candsel}
Given a graph $G=(V,E)$ with a simple topological drawing $\Gamma$ and an integer $k$, can we find a set $S \subset V$ of at most $k$ vertices such that the drawing $\Gamma$ restricted to $G[V\setminus S]$ is planar?%
\end{problem}

Problem~\ref{pb:candsel} is closely related to the \NP-complete problem \textsc{Vertex Splitting}~\cite{marx_obtaining_2012,ly-nphpn-80,kawarabayashi_planarity_2009,jansen_nearoptimal_2013a}, yet it deals with deleting vertices from an arbitrary given drawing of a graph with crossings. One can easily see that Problem~\ref{pb:candsel} is \FPT, using a bounded search tree approach, where for up to $k$ times we select a remaining crossing and branch over the four possibilities of deleting a vertex incident to the crossing edges.
The vertices split in a solution of Problem~\ref{pb:esn} necessarily are a solution to Problem~\ref{pb:candsel}; 
otherwise some crossings would remain in $\Gamma$ after splitting and re-embedding. However, a set corresponding to a minimum-split solution of Problem~\ref{pb:esn} is not necessarily a minimum cardinality vertex deletion set as vertices can be split multiple times.
Moreover, an optimal solution to Problem~\ref{pb:esn} may also split vertices that are not incident to any crossed edge and thus do not belong to an inclusion-minimal vertex deletion set. %
We note here that a solution to Problem~\ref{pb:candsel} solves a problem variation where rather than minimizing the number of splits required to reach planarity, we instead minimize the number of split vertices: Splitting each vertex in an inclusion-minimal vertex deletion set its degree many times trivially results in a planar graph.%

In the re-embedding problem, a graph drawing and a set of candidate vertices to be split %
are given. 
The task is to decide how many times to split each candidate vertex, where to re-embed each copy, and to which neighbors of the original candidate vertex to connect each copy.

\newcommand\pSSRE{\textsc{Split Set Re-Embedding}}

\begin{problem}[\pSSRE]\label{pb:ssre}
Given a graph $G=(V,E)$, a candidate set $S \subset V$ such that %
$G[V \setminus S]$ is planar, a simple planar topological drawing $\Gamma$ of $G[V \setminus S]$,
and an integer $k \ge |S|$, can we perform in $G$ %
at most $k$ splits %
to the vertices in~$S$, where each vertex in~$S$ is split at least once, such that the resulting graph has a planar drawing that coincides with $\Gamma$ on~$G[V \setminus S]$? %
\end{problem}

We note that if no splits were allowed ($k=0$) then Problem~\ref{pb:ssre} would reduce to a partial planar drawing extension problem asking to re-embed each vertex of set~$S$ at a new position without splitting, which can be solved in linear time~\cite{adfjkp-tppeg-15}.

\section{Algorithms for (Embedded) Splitting Number}\label{sec:algorithms}
\appendixsection{sec:algorithms}

\textsc{Splitting Number} is known to be \NP-complete in non-embedded graphs~\cite{Faria_2001}.
In \cref{sec:nonunifpt}, we show that it is \FPT~when parameterized by the number of allowed split operations.
Indeed, we will show something more general, namely, that we can replace planar graphs by any class of graphs that is closed under taking minors and still get an \FPT~algorithm.
Essentially we will show that the class of graphs that can be made planar by at most $k$ splitting operations is closed under taking minors and then apply a result of Robertson and Seymour that asserts that membership in such a class can be checked efficiently~\cite{parameterized_alg}.

\looseness=-1
For vertex splitting in graph drawings, we consider in \cref{sec:poly-alg} the restricted problem to split a single vertex. We show that selecting such a vertex and re-embedding at most $k$ copies of it, while minimizing the number of crossings, can be done in polynomial time for constant~$k$.
For details see \ifArxiv Appendix~\ref{appsec:sec:algorithms} \else the full paper~\cite{arxiv_full_vdup} \fi

\subsection{A Non-Uniform Algorithm for Splitting Number}\label{sec:nonunifpt}

\toappendix
{\subsection{A Non-Uniform Algorithm for Splitting Number}}

We use the following terminology.
Let $G$ be a graph.
A \emph{minor} of $G$ is a graph~$H$ obtained from a subgraph of $G$ by a series of edge contractions.
\emph{Contracting} an edge $uv$ means to remove $u$ and $v$ from the graph, and to add a vertex that is adjacent to all previous neighbors of $u$ and~$v$. %
A graph class $\Pi$ is \emph{minor closed} if for every graph $G \in \Pi$ and each minor $H$ of $G$ we have $H \in \Pi$.
Let $\Pi$ be a graph class and $k \in \mathbb{N}$.
We define the graph class $\Pi_k$ to contain each graph $G$ such that a graph in $\Pi$ can be obtained from $G$ by at most $k$ vertex splits.

\toappendix{
In the next proof we use the concept of a neighborhood cover. For an integer $k \ge 2$, a \emph{neighborhood $k$-cover} of a vertex $v \in V$ is a $k$-tuple $(N_1, \dots, N_k)$ with $N_1,  \dots, N_k \subseteq V$ such that $N_1 \cup \dots \cup N_k = N_G(v)$.
Thus splitting $v$ exactly $k-1$ times results in $k$ copies whose neighborhoods form a neighborhood $k$-cover of~$v$.
}

\begin{restatable}{theorem}{minorclosednessthm}
  \label{lma1}
  For a minor-closed graph class~$\Pi$ and $k \in \mathbb{N}$, $\Pi_k$ is minor closed.%
\end{restatable}%
\appendixproofwithrestatable{\minorclosednessthm*}
{
  \begin{proof}
    Let $G \in \Pi_k$ and let $H$ be a minor of $G$. 
    We show that $H \in \Pi_k$.
    Let $G'$ be a subgraph of $G$ such that $H$ is obtained from $G'$ by a series of edge contractions.
    We first show that $G' \in \Pi_k$.
    Let $s_1, s_2, \ldots, s_{k'}$ be a sequence of at most $k$ vertex splits that, when successively applied to $G$, we obtain a graph in $\Pi$.
    Let $G_0 = G$ and for each $i = 1, \ldots, k'$ let $G_i$ be the graph obtained after applying $s_i$.
    We adapt the sequence $s_1, \ldots, s_{k'}$ to obtain a sequence of graphs $G' = G'_0, G'_1, \ldots, G'_{k'}$ as follows.
    For each $i = 1, 2, \ldots, k'$, if $s_i$ is applied to a vertex $v \in V(G_{i - 1})$ with partition $N_1, N_2$ of $N_{G_{i - 1}}(v)$ then, if $v \in V(G'_{i - 1})$, to get $G'_i$ we apply a split operation in $G'_{i - 1}$ to $v$ with neighborhood cover $(N_1 \cap V(G'_{i - 1}), N_2 \cap V(G'_{i - 1}))$ (and we assume without loss of generality that the new vertices introduced by this operation are identical to the vertices introduced by $s_i$).
    If $v \notin V(G'_{i - 1})$ we put $G_i = G_{i - 1}$ instead.
    Observe that for each $i \in \{1, 2, \ldots, k'\}$, graph $G'_i$ is a subgraph of $G_i$.
    Hence, $G'_{k'}$ is a subgraph of $G_{k'}$ and since $\Pi$ is in particular closed under taking subgraphs, we have $G'_{k'} \in \Pi$, as claimed.

    It remains to show that for each graph $H$ that is obtained from a graph $G \in \Pi_k$ through a series of edge contractions we have $H \in \Pi_k$.
    By induction on the number of edge contractions, it is enough to show this in the restricted case where $H$ is obtained by a single edge contraction from $G$.
    Let $s_i$ and $G_i$ be as defined above, that is, $s_1, \ldots, s_{k'}$ is a sequence of vertex-split operations that when successively applied to $G$, we obtain a graph in $\Pi$ and $G_i$ is the graph obained after applying~$s_i$.
    We claim that there is a series of $k'$ split operations applied to $H = H_0$ that result in a series of graphs $H_1, H_2, \ldots, H_{k'}$ such that for all $i \in \{1, 2, \ldots, k'\}$ a graph isomorphic to $H_i$ is obtained from $G_i$ by a single edge contraction.
    This would imply that $H_{k'} \in \Pi$ because $\Pi$ is minor closed and thus we would have $H \in \Pi_k$, finishing the proof.

    To prove the claim, since $H_0$ is obtained from $G_0$ by a single edge contraction, by induction it is enough to show the following.
    Let $H=(V',E')$ be obtained from $G=(V,E)$ by contracting the single edge $uv \in E$ and let $w \in V'$ be the vertex resulting from the contraction.
    Let $G'$ be obtained from $G$ by applying a single split operation~$s$.
    It is enough to show that there is a split operation such that, when applied to $H$ to obtain the graph $H'$ we have that $H'$ is isomorphic to a graph obtained from $G'$ by contracting a single edge.
    
    To show this, if $s$ is not applied to $u$ or $v$, then $s$ can directly be applied to $H$.
    Then, contracting $uv$ in $G'$ we see directly that we obtain a graph isomorphic to $H'$, as required.

    \begin{figure}[t]
      \includegraphics[page=2,width=\linewidth]{edgecontraction.pdf}
      \caption{$G$ and $G'$ after a splitting operation on $u$ (top row) and $H$ and $H'$ after a splitting operation on $w$, which is the result of contracting edge $uv$ (bottom row). The set $N_G(u)$ is split into $N_1$ and $N_2$. And edge between a vertex and a set indicates the possibility of the vertex having neighbors in that set.
      }
      \label{fig:edgecont}
    \end{figure}
    
    Otherwise, $s$ is applied to $u$ or $v$.
    By symmetry, say $s$ is applied to $u$ without loss of generality.
    Let thus $u \in V(G)$ be split into $u_1$ and $u_2$ in $G'$.
    We split $w$ into $w_1$ and $w_2$ in $H$ with the neighborhood 2-cover of $w$ defined as $(N_{G'}(u_1), N_{G'}(u_2) \cup N_{G'}(v))$.
    That is, $H'$ is obtained from $H$ by splitting $w$ into $w_1$ and $w_2$ such that $N_{H'}(w_1) = N_{G'}(u_1)$ and $N_{H'}(w_2) = N_{G'}(u_2) \cup N_{G'}(v)$ (see \cref{fig:edgecont}).
    Contracting $u_2v$ in $G'$ into a vertex $x$ results in a graph $\hat{G}$ that is isomorphic to $H'$:
    To see this, observe that all neighborhoods of vertices in $V \setminus \{u_1, x\}$ are identical between $\hat{G}$ and $H'$ and that $N_{\hat{G}}(u_1) = N_{H'}(w_1)$ and $N_{\hat{G}}(x) = N_{G'}(u_1) \cup N_{G'}(v)  = N_{H'}(w_2)$.
    Thus, our claim is proven, finishing the overall proof.
  \end{proof}
}
\noindent
The proof is given in 
\ifArxiv Appendix~\ref{appsec:sec:algorithms} \else the full paper~\cite{arxiv_full_vdup} \fi
and essentially shows that whenever we have a graph $G \in \Pi_k$ and a minor $H$ of $G$, then we can retrace vertex splits in $H$ analogous to the splits that show that $G$ is in $\Pi_k$.
By results of Robertson and Seymour we obtain (again see
\ifArxiv Appendix~\ref{appsec:sec:algorithms}\else\cite{arxiv_full_vdup}\fi):

\begin{restatable}{proposition}{minorclosednessfptthm}
\label{prop:split-to-minor-closed-nonuni-fpt}
  Let $\Pi$ be a minor-closed graph class.
  There is a function $f \colon \mathbb{N} \to \mathbb{N}$ such that for every $k \in \mathbb{N}$ there is an algorithm running in $f(k)\cdot n^3$ time that, given a graph $G$ with $n$ vertices, correctly determines whether $G \in \Pi_k$.
\end{restatable}
\appendixproofwithrestatable{\minorclosednessfptthm*}{
  \begin{proof}
    From \cref{lma1} it follows that the class $\Pi_k$ of graphs that represent positive input instances is closed under taking minors.
    From Robertson and Seymour's graph minor theorem it follows that it can be determined in $c_{\Pi_k} \cdot n^3$ time for a given $n$-vertex graph whether it is contained in $\Pi_k$, where $c_{\Pi_k}$ is a constant depending only on $\Pi_k$ %
    (see Cygan et al.~\cite[Theorem 6.13]{parameterized_alg}).
    \cref{prop:split-to-minor-closed-nonuni-fpt} follows by setting $f(k) = c_{\Pi_k}$.
  \end{proof}
}

Since the class of planar graphs is minor closed, we obtain the following.

\begin{corollary}
  \textsc{Splitting Number} is non-uniformly fixed-parameter tractable\footnote{A parameterized problem is non-uniformly fixed-parameter tractable if there is a function $f \colon \mathbb{N} \to \mathbb{N}$ and a constant $c$ such that for every parameter value $k$ there is an algorithm that decides the problem and runs in $f(k) \cdot n^c$ time on inputs with parameter value $k$ and length~$n$.} with respect to the number of allowed vertex splits.
\end{corollary}

\subsection{Optimally Splitting a Single Vertex in a Graph Drawing}\label{sec:poly-alg}

Let $\Gamma$ be a drawing of a graph $G=(V,E)$ and let $v \in V$ be the single vertex to be split $k$ times. Chimani et al.~\cite{Chimani_2009} showed that inserting a single star into an embedded graph while minimizing the number of crossings can be solved in polynomial time by considering shortest paths between faces in the dual graph, whose length correspond to the edges crossed in the primal. We build on their algorithm by computing the shortest paths in the dual of the planarized subdrawing of $\Gamma$ for $G[V \setminus \{v\}]$ between all faces incident to $N(v)$ and all possible faces for re-inserting the copies of $v$ as the center of a star. We  branch over all combinations of $k$ faces to embed the copies, compute the nearest copy for each neighbor and select the combination that minimizes the number of crossings.

\toappendix
{\subsection{Optimally Splitting a Single Vertex in a Graph Drawing}
Given a graph $G=(V,E)$ and its drawing $\Gamma$, a candidate vertex $v\in V$ and an integer $k$, we show that we can split $v$ into $k$ copies such that the resulting number of crossings is minimized; we construct a corresponding crossing-minimal drawing $\Gamma^*$ in polynomial time.

Chimani et al.~\cite{Chimani_2009} showed that inserting a single star into an embedded graph while minimizing the number of crossings can be solved in polynomial time. They use the fact that the length of a shortest path in the dual graph between two vertices $v_1$ and $v_2$ corresponds to the number of edge crossings generated by an edge between two vertices embedded on the faces in the primal graph corresponding to $v_1$ and $v_2$.
We extend this method to optimally split a single vertex into $k$ copies, which is similar to reinserting $k$ stars. The algorithm planarizes the input graph, then exhaustively tries all combinations of $k$ faces to re-embed the copies of $v$. For each combination it finds for all the neighbors of $v$ which copy is their best new neighbor (meaning it induces the least amount of crossings) using the dual graph. 

In the first step we remove the specified vertex $v$ from $\Gamma$ with all its incident edges.
Let $\Pi$ be the planarization of $\Gamma\setminus v$, %
let $F$ %
be the set of faces of $\Pi$, and let 
$D$ be the dual graph of $\Pi$.
For a vertex $u\in\Pi$ incident to a face set $F(u)$, we define $V_F(u)$ as the vertex set that represents $F(u)$ in $D$.
The algorithm by Chimani et al.~\cite{Chimani_2009} computes the crossing number for the vertex insertion in a face $f\in F$ by finding the sum of the shortest paths in $D$ between the vertex $v_f$ that represents $f$ in $D$ and $V_F(w)$, for each $w\in N(v)$.
In our algorithm, we collect all the individual path lengths in $D$ between each face vertex~$v_f$ and the faces in $V_F(w)$, for each $w\in N(v)$, in a table, and then consider all face subsets of size $k$ in which to embed the $k$ copies of $v$. 
For such a given set $S$ of $k$ faces, we assign each $w\in N(v)$ to its closest face $f^*(w) \in S$, which yields a crossing-minimal edge between $w$ and $S$. %
We break ties in path lengths lexicographically using a fixed order of $F$. 
For each set $S$ of candidate faces we compute the sum of the resulting shortest path lengths between $w$ via some face in $V_F(w)$ and $f^*(w)$ for each $w \in N(v)$.

We choose as the solution the set $S^*$ with minimum total path length and assign one copy $\dot{v}^{(1)},\dots,\dot{v}^{(k)}$ of $v$ into each face
of $S^*$ that is closest to at least one of the neighbors in $N(v)$. 
This corresponds to splitting $v$ into at most $k$ copies.
The edges from the newly placed copies to the neighbors of $v$ follow the computed shortest paths in $D$.

Chimani et al.~\cite{Chimani_2009} showed that we can compute the table of path lengths in $D$ in time $O((|F|+|\mathcal{E}|) |N(v)|)$, where $F$ and $\mathcal{E}$ are, respectively, the sets of faces and edges of~$\Pi$. We consider $O(|F|^k)$ subsets of $k$ faces and chose the one minimizing crossings. Thus our algorithm runs in polynomial time for $k \in O(1)$.
}

\begin{restatable}{theorem}{onevtxcrossmin}
\label{crossmin}
Given a drawing $\Gamma$ of a graph $G$, a vertex $v \in V(G)$, and an integer $k$, we can split $v$ into $k$ copies such that the remaining number of crossings is minimized in time $O((|F|+|\mathcal{E}|) \cdot |N(v)| \cdot |F|^k)$, where $F$ and~$\mathcal{E}$ are respectively the sets of faces and edges of the planarization of $\Gamma$.
\end{restatable}

\appendixproofwithrestatable{\onevtxcrossmin*}{
\begin{proof}
Given a solution that embeds copies in the faces $f_1,\dots,f_k$, the drawing we compute has a minimum number of edge crossings. Since for each face combination the algorithm finds the minimum number of crossings with the edges of $\Gamma$ of a vertex insertion, by design, the star centered at each copy has minimum number of crossings with $\Gamma$ by the exhaustive search. Still, uncounted crossings between inserted edges could happen. We argue that this is impossible. Let us assume that we have two stars rooted on $v_i$, $v_j$, and vertices $u_i,u_j\in {N(v)}$ s.t.\ $(v_i,u_i)$ crosses the edge $(v_j,u_j)$ in a point $c$ inside some face~$f$. This means that in the dual graph both paths representing the two edges pass through the same vertex $v_f$ that represents $f$. If the path between $c$ and $v_i$ crosses less edges than the path between $c$ and $v_j$ then the path in $D$ from $v_j$ to $u_j$ is not minimal and $u_j$ would have been assigned to $v_i$ (symmetrically if the path between $c$ and $v_j$ crosses less edges). If both paths have the same length, then by the lexicographic order rule, both vertices would have been assigned to the same face. So the stars do not intersect one another and we search exhaustively for all possible face sets to embed the stars, meaning that the algorithm produces a crossing minimal drawing after splitting $v$ at most $k$ times.
Finally, to find the amount of crossings generated by adding an edge between a neighbor $w$ of $v$, we can do a BFS traversal of the graph from the faces of $V_F(w)$, and every time we encounter a face, the depth on the tree corresponds to the face distance to $V_F(w)$. We do this for every element of $N(v)$, and for every set of $k$ faces of the planarization of $\Gamma$, which takes $O((|F|+|\mathcal{E}|) \cdot |N(v)| \cdot |F|^k)$ time.
\end{proof}
}

\section{\NP-completeness of %
Subproblems}\label{sec:complex}

While it is known that \sn\ is \NP-complete~\cite{Faria_2001}, in the correctness proof of the reduction Faria et al.~\cite{Faria_2001} assume that it is permissible to draw all vertices, split or not, at new positions as there is no initial drawing to be preserved.
The reduction thus does not seem to easily extend to \textsc{Embedded Splitting Number}.
Here we show the \NP-completeness of each of its two subproblems. 

\begin{figure}[tb]
	\centering
	\includegraphics{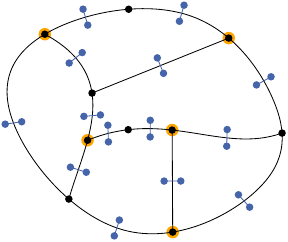}
	\caption{
	The drawing $\Gamma$ in black, and the vertices and edges added to obtain $\Gamma'$ in blue. The vertex cover highlighted in orange corresponds to the deletion set.}
	\label{fig:newest_reduction_overview}
\end{figure}

\begin{restatable}{theorem}{hardnessone}
	\label{thm:hardnessone}
	\pCS\ is \NP-complete. %
\end{restatable}

\begin{proof}

We reduce from the \NP-complete \textsc{Vertex Cover} problem in planar graphs~\cite{GareyJ77}, where given a planar graph $G=(V,E)$ and an integer $k$, the task is to decide if there is a subset $V' \subseteq V$ with $|V'| \le k$ such that each edge $e \in E$ has an endpoint in $V'$.
Given the planar graph $G$ from such a \textsc{Vertex Cover} instance and an arbitary plane drawing $\Gamma$ %
of $G$ we construct an instance of \pCS\ as follows.
We create a drawing $\Gamma'$ by drawing a \emph{crossing edge} $e'$ across each edge $e$ of $\Gamma$ such that $e'$
is orthogonal to 
$e$ and has a small enough positive length such that
$e'$
intersects only $e$
and no other crossing edge or edge in $\Gamma$, see \cref{fig:newest_reduction_overview}.
Drawing $\Gamma'$ can be computed in polynomial time.

Let $C$ be a vertex cover of $G$ with $|C| = k$.  We claim that $C$ is also a deletion set that solves \pCS\ for $\Gamma'$. 
We remove the vertices in $C$ from $\Gamma'$, with their incident edges. 
By definition of a vertex cover, this removes all the edges of $G$ from $\Gamma'$. 
The remaining edges in $\Gamma'$ are the crossing edges %
and they form together a (disconnected) planar drawing which shows that $C$ is a solution of \pCS\ for $\Gamma'$.

Let $D$ be a deletion set of $\Gamma'$ such that $|D|=k$. We find a vertex cover of size at most $k$ for $G$ in the following manner. Assume that $D$ contains a vertex $w$ that is an endpoint of a crossing edge $e$ that crosses the edge $(u,v)$ of $G$. Since $w$ has degree one, deleting it only resolves the crossing between $e$ and $(u,v)$, thus we can replace $w$ in $D$ by $u$ (or $v$) and resolve the same crossing as well as all the crossings induced by the edges incident to $u$ (or $v$). Thus we can find a deletion set $D'$ of size smaller or equal to $k$ that contains only vertices in $G$ and %
removing this deletion set from $\Gamma'$ removes only edges from $G$. 
Since every edge of $G$ is crossed in $\Gamma'$, every edge of $G$ must have an incident vertex in $D'$, thus $D'$ is a vertex cover for $G$.

Containment in \NP\ 
is easy to see. Given a deletion set $D$, we only need to verify that $\Gamma'$ is planar after deleting $D$ and its incident edges.
\end{proof}

Next, we prove that also the re-embedding subproblem itself is \NP-complete, by showing that \textsc{Face Cover} is a special case of the re-embedding problem. %
The problem \textsc{Face Cover} is defined as follows. Given a planar graph $G=(V,E)$, a subset~$D\subseteq V$, and an integer $k$, can $G$ be embedded in the plane, such that at most $k$ faces are required to cover all the vertices in $D$? \textsc{Face Cover} is \NP-complete, even when $G$ has a unique planar embedding~\cite{bm-ccvfpg-88}. %

\begin{restatable}{theorem}{hardnesstwo}\label{thm:hardnesstwo}
	\pSSRE\ is \NP-complete.
\end{restatable}
\begin{proof}
	We %
	give a parameterized reduction from \textsc{Face Cover} (with unique planar embedding) parameterized by the solution size to the re-embedding problem parameterized by the number of allowed splits.
	We first create a graph~$G' = (V', E')$, with a new vertex $v$, vertex set $V' = V \cup \{v\}$ and $E' = E \cup \{dv \mid d \in D\}$.
	Then we compute a planar drawing $\Gamma$ of $G$ corresponding to the unique embedding of $G$. %
	Finally, we define the candidate set $S=\{v\}$ and allow for $k-1$ splits in order to create up to $k$ copies of $v$. %
	Then $G'$, $\Gamma$, $S$, and $k-1$ form an instance $I$ of \pSSRE. %
	
	In a solution of $I$, every vertex in $D$ is incident to a face in~$\Gamma$, in which a copy of $v$ was placed.
	Therefore, selecting these at most $k$ faces in~$\Gamma$ gives a solution for the \textsc{Face Cover} instance.
	Conversely, given a solution for the \textsc{Face Cover} instance, we know that every vertex in $D$ is incident to at least one of the at most $k$ faces.
	Therefore, placing a copy of $v$ in every face of the \textsc{Face Cover} solution yields a re-embedding of at most $k$ copies of $v$, each of which can realize all its edges to neighbors on the boundary of the face without crossings. %
	
	Finally, a planar embedding of the graph can be represented combinatorially in polynomial space. %
    We can also verify in polynomial time that this embedding is planar and exactly the right connections are realized, for \NP-containment.%
\end{proof}

\section{Split Set Re-Embedding is Fixed-Parameter Tractable}\label{sec:re-embed-fpt}

\begin{figure}[t]
	\centering
	\includegraphics{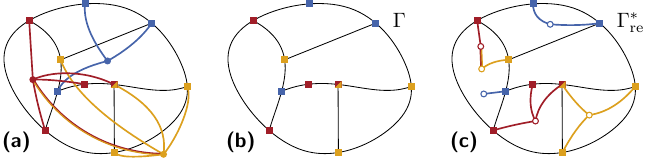}
	\caption{
		\textbf{\textsf{(a)}} An example graph~$G$, \textbf{\textsf{(b)}} a planar drawing~$\Gamma$ of~$G$ where $\candidate$ has been removed, and \textbf{\textsf{(c)}} a solution drawing $\planarreemb^*$. Pistils are squares, copies are circles and vertices in $\candidate$ are disks.
	}
	\label{fig:steps}
\end{figure}

In this section we show that \pSSRE\ (\textsc{SSRE}) can be solved by an \FPT-algorithm, with the number~$k$ of splits as a parameter.
We provide an overview of the involved techniques and algorithms in this section and refer to \ifArxiv Appendix~\ref{appsec:sec:re-embed-fpt} \else\cite{arxiv_full_vdup} \fi
for the full technical details.

\looseness=-1
\subsubsection{Preparation.}
First, from the given set $\candidate$ of $s= |\candidate|$ candidate vertices (disks in \cref{fig:steps}(a)) we choose how many copies of each candidate we will insert back into the graph; these copies form a set~$\ssplit$. 
Vertices with neighbors in $\candidate$ are called \emph{pistils} (squares in \cref{fig:steps}), and faces incident to pistils are called \emph{petals}.
The copies in~$\ssplit$ must be made adjacent to the corresponding pistils. 
Since vertices in $\candidate$ can be pistils, we also determine which of their copies in~$\ssplit$ are connected by edges. For both of these choices, the number of copies per candidate and the edges between copies, we branch over all options\ifArxiv ~(see Appendix~\ref{ssec:10kouter}). \else. \fi
In each of the $\bigoh(2^k k^4)$ resulting branches we apply dynamic programming to solve \textsc{SSRE}.

\looseness=-1
Second, we prepare the drawing~$\Gamma$ for the dynamic programming.
A face $f$ is not necessarily involved in a solution, e.g., if it is not a petal: a copy embedded in $f$ either has no neighbors and can be re-embedded in any face, or its neighbors are not incident with $f$, and this embedding induces a crossing. Therefore we remove all vertices not incident to petals, which actually results in a drawing of a $10k$-outerplanar graph\ifArxiv ~(see Appendix~\ref{ssec:10kouter}). \else. \fi We show this in the following way. We find for each vertex a \emph{face path} to the outerface, which is a path alternating between incident vertices and faces. Each face visited in that path is either a petal or adjacent to a petal (from the above reduction rule). Thus, each face in the path might either have a copy embedded in it in the solution, or a face up to two hops away has a copy embedded in it. If we label each face by a closest copy with respect to the length of the face path, then one can show that no label can appear more than five times on any shortest face path to the outerface. Since there are at most $2k$ different copies (labels), we can bound the maximum distance to the outerface by $10k$.
Because of the $10k$-outerplanarity, we can now compute a special branch decomposition of our adapted drawing~$\Gamma$ in polynomial time, a so called \emph{sphere-cut decomposition} of branchwidth $20k$.
A sphere-cut decomposition $(\lambda,T)$ is a tree $T$ and a bijection $\lambda$ that maps the edges of $\Gamma$ to the leaves of $T$ (see \cref{fig:scd1}).%

\begin{figure}
	\centering
	\includegraphics{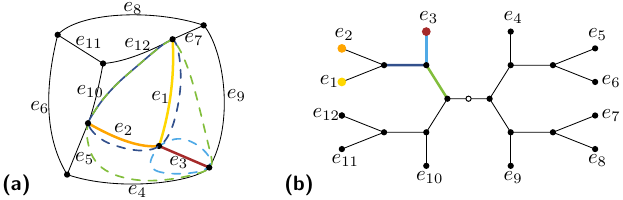}
	\caption{
		\textbf{\textsf{(a)}} A graph and \textbf{\textsf{(b)}} its sphere-cut decomposition. Each labeled leaf corresponds to the same labeled edge of the graph. The middle set of each colored edge in the tree corresponds to the vertices of the corresponding colored dashed noose in the graph.
	}
	\label{fig:scd1}
\end{figure}
\looseness=-1
Each edge~$e$ of $T$ splits $T$ into two subtrees, which induces via $\lambda$ a bipartition $(A_e,B_e)$ of $E(G)$ into two subgraphs. We define a vertex set $\midset(e)$, which contains all vertices that are incident to an edge in both $A_e$ and $B_e$.
Additionally, $e$ also corresponds to a curve, called a \emph{noose} $\eta(e)$, that intersects $\Gamma$ only in $\midset(e)$. 
We define a root for~$T$, and we say that for edge~$e$, the subtree further from the root corresponds to the drawing \emph{inside} the noose $\eta(e)$. Having a root for~$T$ allows us to solve \textsc{SSRE} on increasingly larger subgraphs in a structured way, starting with the leaves of $T$ (the edges of $\Gamma$) and continuing bottom-up to the root (the complete drawing~$\Gamma$). 
\ifArxiv For more details see Appendix~\ref{sec:ssre-fpt-scbd}.\fi
\looseness=-1

\subsubsection{Initialization.}
During the dynamic programming, we want to determine whether there is a partial solution for the subgraphs of $\Gamma$ that we encounter when traversing~$T$. 
For one such subgraph we describe such a \emph{partial solution} with a tuple $(\ssplitp, (N_{\dot{v}})_{\dot{v} \in \ssplitp}, \Gamma')$. In this tuple, the set $\ssplitp \subseteq \ssplit$ corresponds to the embedded copies, $(N_{\dot{v}})$ are their respective neighborhoods in $\Gamma$, and $\Gamma'$ is the resulting drawing.
However, during the dynamic programming, we need more information to determine whether a partial solution exists. 
For example, for the faces completely inside the %
noose~$\eta$ enclosing~$\Gamma'$ (\emph{processed} faces) a solution must already be found, while faces intersected by~$\eta$ (\emph{current} faces) still need to be considered.

We store this information in a \emph{signature}, which is a tuple $(\inner, \cycle, M, \noose)$. An example of a signature is visualized in~\cref{fig:sigcontent}.
The set $\inner(t)\subseteq \ssplit$ corresponds to the copies embedded in processed faces, 
and $\noose$ contains a set~$X(p)$ for each pistil~$p$ on the noose, such that~$X(p)$ describes which neighbors of $p$ are still missing in the partial solution.
The set $\cycle$ corresponds to a set of planar graphs describing the combinatorial embedding of copies in current faces. One such graph $C_f$ (\cref{fig:nesting}) associated with a current face $f$ consists of a cycle, whose vertices represent the pistils of $f$, and of copies embedded inside the cycle. Since $f$ is current, not all of its pistils are necessarily inside the noose, and $M$ describes which section of the cycle, and hence which pistils, should be used.

\begin{figure}[tb]
  \centering
  \includegraphics{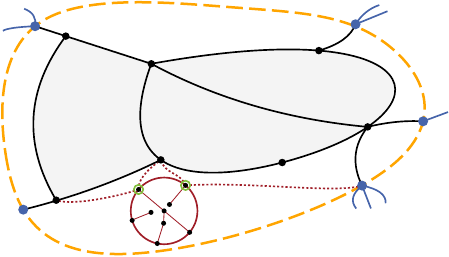}
  \caption{The information stored in a signature of the partial solution inside the orange dashed noose: copies in $\inner$ are used in the grey faces, blue noose vertices who have missing neighbors (outgoing edges outside the noose) are stored in $\noose$, in red an example of a nesting graph for a face traversed by the noose, with four dotted edges connecting to the cycle and the vertices described by $M$ in green.}
  \label{fig:sigcontent}
\end{figure}

\looseness=-1
Saving a single local optimal partial solution, one that uses the smallest number of copies, for a given noose is not sufficient. This sub-solution may result in a no-instance when considering the rest of the graph outside the noose.
We therefore keep track of all signatures that lead to partial solutions, which we call \emph{valid} signatures. These signatures allow us to realize the required neighborhoods for pistils inside the noose with a crossing-free drawing. The number of distinct signatures~$N_s(k)$ depends on the number $k$ of splits and we prove an upper bound of $2^{O(k^2)}$ by counting all options for each element of a signature tuple. Since the number of signatures is bounded by a function of our parameter~$k$, we can safely enumerate all signatures.
We then determine which signatures are valid for each noose in $T$\ifArxiv ~(see Appendix~\ref{sec:ssre-fpt-dp}). \else. \fi
The number~$N_s(k)$ of signatures will be part of the leading term in the total running time.
\looseness=-1

\begin{figure}[tb]
	\centering
	\includegraphics{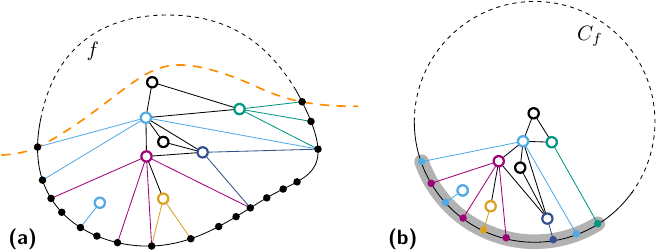}
	\caption{
          \textbf{\textsf{(a)}} A face~$f$ and copies inside the orange noose, and \textbf{\textsf{(b)}} the corresponding nesting graph~$C_f$ with the interval described by $M$ highlighted in grey. The two light blue vertices represent two different copies of the same removed vertex. Copies in $C_f$ have no edges to copies in other nesting graphs.}
          \label{fig:nesting}
\end{figure}

\toappendix{

\section{Split Set Re-Embedding is Fixed-Parameter Tractable}\label{appsec:sec:re-embed-fpt}

\looseness=-1
We first introduce the following terminology. Any vertex $v$ in $\Gamma$ that has a neighbor in $\emph{\candidate}$ is called a \emph{pistil}.
Each face that is incident to a pistil is called a \emph{petal}.
Let $p$ be a pistil in the input graph $G$ with neighbors $N(p)$. 
Let $\dot{v}$ be a copy of some $v \in \candidate$, where $v \in N(p)$. We say $\dot{v}$ \emph{covers} $p$ if it is embedded in a face incident to $p$ and we can draw a crossing free edge between $\dot{v}$ and $p$.

If $(G, \Gamma, \candidate, k)$ is a yes-instance of \pSSRE, then there is a series of split operations to achieve a planar re-embedding, we call the elements of the series \emph{solution splits} (see \cref{fig:steps} for an example of the problem and its solutions).
We refer to the \emph{solution graph} as the graph 
obtained from $G$ by performing the solution splits.
A \emph{solution} is defined as a tuple 
$(\ssplitsol, \copiessol, \origsol, (N_u^*)_{u \in \ssplitsol}, \planarreemb^*)$ 
consisting of the following:

\begin{enumerate}[(i)]
\item The set $\ssplitsol$ of copies of vertices in $\candidate$ introduced by the solution splits.
  Since $\candidate$ has size $s \leq k$ and there are at most $k$ splits, we have $|\ssplitsol| \leq s + k \leq 2k$.
\item A mapping $\copiessol \colon \candidate \to 2^\ssplitsol$ that maps each vertex $v \in \candidate$ to the set $\copiessol(v) \subseteq \ssplitsol$ of copies of $v$ introduced by performing the solution splits.
\item A mapping $\origsol \colon \ssplitsol \to \candidate$ that maps each copy $\dot{v} \in \ssplitsol$ to the vertex $v=\origsol(\dot{v}) \in \candidate$ that $\dot{v}$ is a copy of.
\item For each copy $\dot{v} \in \ssplitsol$ a vertex set $N_{\dot{v}}^* \subseteq V$ of neighbors of $\dot{v}$ such that for each $v \in \candidate$ the family $\{ N_{\dot{v}}^* \mid \dot{v} \in \copiessol(v) \}$ is a partition of $N_G(v)$.
\item A planar drawing $\planarreemb^*$ of the graph resulting from $\Gamma$ by embedding the copies in $\ssplitsol$ such that each copy $\dot{v} \in \ssplitsol$ has edges drawn to  each vertex in $N_{\dot{v}}^*$ (\cref{fig:steps}c).
\end{enumerate}

\subsection{Splitting Candidates and Obtaining $10k$-Outerplanarity}\label{ssec:10kouter}

As mentioned, the first step in our algorithm for \pSSRE\ is to determine, (a) for each candidate vertex in $\candidate$ how many copies are introduced and (b)~how all these copies are connected to each other.
For (a), we branch into all possibilities of performing at most $k$ splits of the candidate vertices in $\candidate$; in each branch we carry out all steps of the algorithm described below.
Since there are at most $k$ splits, at least one per candidate, and $s \leq k$, we are looking for an $s$-composition of the integer $k$ which is given by $\binom{k-1}{s-1}\leq 2^k$.
To keep track of the splits done in the current branch, we define $\ssplit$ as the set of resulting copies of $\candidate$ and we define the corresponding mappings $\orig$ and $\copies$.
Note that $|\ssplit| \leq 2k$.
After branching, we look only for solutions that respect the choice of vertex splits in each branch.
Clearly, if there is a solution for the instance, one branch will have made the same choice.
Analogously, for (b) we create a branch for each of the $2^{O(k^{2})}$ possible graphs with vertex set~$\ssplit$ that represent $G[\ssplit]$.
We denote the choice of this graph in the current branch by~$G_\ssplit$.
Similar as for the previous branching, we look only for solutions that respect the choice of a certain branch, i.e., we require that the set of copies induces the subgraph $G_{\ssplit}$.

\subsubsection{Initializing the Split Vertex Set \boldmath$\ssplit$}\label{sec:ssre-fpt-split-guess-1}
As mentioned, the first step in our algorithm for \pSSRE\ is to determine, for each candidate split vertex, how many copies are introduced and how these copies are connected to each other.  
We are given as input the set $\candidate$ of vertices that have been removed from the drawing to be split such that $|\candidate|=s\leq k$. Otherwise, if $|\candidate|=s > k$, we can immediately conclude that we deal with a no-instance as all candidate vertices must be split at least once.
From $\candidate$ we now create sets $\ssplit$ of copies that will be reintroduced into the drawing.
To obtain all relevant sets $\ssplit$, we introduce our first branching rule.

\begin{branchingR}
\label{br:vertex}
  Let $(G, \Gamma, \candidate, k)$ be an instance of $\pSSRE$. Create a branch for every $k' = s, \ldots, s + k$, defining a set $\ssplit$ of $k'$ new vertices, for every mapping $\orig \colon \ssplit \to \candidate$, and every mapping $\copies \colon \candidate \to 2^\ssplit$ such that $\{\copies(v) \mid v \in \candidate\}$ is a partition of $\ssplit$ and such that $\orig^{-1} = \copies$.\footnote{Herein for a mapping $f \colon X \to Y$ we define the mapping $f^{-1} \colon Y \to 2^X$ by putting $f^{-1}(y) = \{x \in X \mid f(x) = y\}$ for every $y \in Y$.}
\end{branchingR}

Note that for resolving crossings in $\Gamma$ one might sometimes want to only re-embed a vertex without splitting. However, such a vertex move operation is not permitted in \pESN\ and we only permit re-embedding copies of original vertices. Thus each vertex in $\candidate$ is split at least once.  
We can bound the number of branches by observing that the number of ways of distributing the $k$ splits among $s$ vertices with at least one split per vertex is the same as the number of $s$-compositions of the integer $k$, which is described by $\binom{k-1}{s-1} \le 2^k$.

\begin{lemma}\label{lma:brvertex}
  If $(G, \Gamma, \candidate, k)$ is a yes-instance of $\pSSRE$ with solution %
  $\sigma=(\ssplitsol, \copiessol, \origsol, (N_{\dot{v}})_{\dot{v} \in \ssplitsol}, \planarreemb^*)$, then there is a branch created by Branching Rule~\ref{br:vertex} with $\ssplitsol = \ssplit$.
\end{lemma}
\begin{proof}
Given the set $\ssplitsol$ of copies re-embedded by a solution, we argue that $\ssplitsol$ must be a subset of a set generated in at least one of the branches of BR1.
Set $\ssplitsol$ necessarily contains at least one copy of each vertex in $\candidate$, and we call $\addit$ the remainder of copies embedded by $\ssplitsol$.
Since the total number of splits is bounded by $k$, $\addit$ contains at most $k=s+d$ vertices: one copy of each vertex in $\candidate$ for the initial split, and $d$ vertices for additional splits. Our branching rule BR1 creates a branch for each combination of copies from $\candidate$ that adds up to $2s+d$, two copies of each vertex in $\candidate$, and one branch for each combination of $d$ copies of $\candidate$. Thus there is a branch~$b$ that chooses a set of copies $S_b$ that contains a superset of~$\addit$, such that $S_b \setminus \addit$ contains exactly one copy of each vertex in $\candidate$. Thus $S_b$ contains exactly all copies $\ssplitsol$ re-embedded by a solution.
\end{proof}

Finding a superset of an optimal set of copies is sufficient, as having additional copies of some vertices does not prevent us from finding a solution.

Some vertices in $\ssplit$ have neighbors in $\ssplit$, or in both $\ssplit$ and $V \setminus \candidate$. From the perspective of a non-split pistil it is easy to verify whether the original neighborhood is present in a solution. However, a split vertex is incident to only a subset of its original edges. Thus it is required to consider the union of all copies of a vertex to find its original neighborhood. For an edge $(u,v)$ where $u,v\in \candidate$, we will branch over all possibilities for defining this edge between any pair of copies of $u$ and $v$. So if there are $i$ copies of $u$ and $j$ copies of $v$, we will branch over the $i\cdot j$ possible ways to reflect $(u,v)$ in the solution. 
This simplifies the verification of the presence of such edges as it is sufficient to consider the neighborhood of any chosen copy of a vertex, rather than of all its copies.

\begin{branchingR}
\label{br:edge}
Let $(G=(V,E), \Gamma, \candidate, k)$ be an instance of \pSSRE\ and $\ssplit$ a set of copies obtained from Branching Rule~\ref{br:vertex}.
Create a branch for each possible set $E_{\splt}=\{ (\dot{u}^{(i)},\dot{v}^{(j)}) \mid \dot{u}^{(i)},\dot{v}^{(j)}\in\ssplit \text{ and }$ $(\origsol(\dot{u}^{(i)}),\origsol(\dot{v}^{(j)}))\in E\}$, in which every $e\in E$ is represented exactly once. %
\end{branchingR}

Let $s \ge 2$ as otherwise there are no such edges within $\ssplit$. Each of the $s \le k$ vertices in $\candidate$ has at most $k-1$ neighbors in $\candidate$, yielding at most $\binom{k}{2}$ edges. Since each vertex in $\candidate$ has at most $k$ copies in $\ssplit$, there are at most $k^2$
possibilities to reflect one original edge in $\binom{\candidate}{2}$ as an edge in $\binom{\ssplit}{2}$. This results in $O(k^4)$ branches. %

\begin{lemma}\label{lma:bredge}
  If $(G=(V,E), \Gamma, \candidate, k)$ an instance of $\pSSRE$ is a yes-instance with solution $\sigma=(\ssplitsol, \copiessol, \origsol, (N_{\dot{v}})_{{\dot{v}} \in \ssplitsol}, \planarreemb^*)$ and graph $G^*$ corresponding to $\planarreemb^*$, then there is a branch created by Branching Rule~\ref{br:edge} with $E_{\splt} = E(G^*[\ssplitsol]$).

\end{lemma}
\begin{proof}
  By Lemma~\ref{lma:brvertex} we know that Branching Rule~\ref{br:vertex} produces a branch that passes $\ssplitsol$ to Branching Rule~\ref{br:edge}. Therefore, in this branch the edges in $E_{\splt}$ must have both endpoints in $\ssplitsol$ in the solution. Let $\dot{u}^{(i)},\dot{v}^{(j)} \in \ssplitsol$ such that  $(\origsol(\dot{u}^{(i)}),\origsol(\dot{v}^{(j)}))\in E$.
  Let us assume no branch created by Branching Rule 2 contains $(\dot{u}^{(i)},\dot{v}^{(j)})$ and we do not find the solution. 
  As the branching rule attempts all possible assignments of copy pairs, it necessarily attempts to add $(\dot{u}^{(i)},\dot{u}^{(j)})$ and the assumption contradicts the branching rule. So the branching rule is sound.
\end{proof}

As a second step, we want to reduce the size of our input, or more precisely the input drawing $\Gamma$.
Note that faces that are not incident to any pistil do not play an important part in the solution as there are no crossing-free edges that can be realized in such a face to achieve our desired adjacencies between copies and pistils.
Hence, we use a reduction rule that removes all vertices not adjacent to a petal from~$\Gamma$.
We now show that this results in a $10k$-outerplanar drawing.
\subsubsection{Obtaining $10k$-Outerplanarity}\label{sec:ssre-fpt-outerplanar}

Our algorithm in this section should output a planar drawing. 
This means that a copy $\dot{v} \in \ssplitsol$ that is reinserted into $\Gamma$ is added only into those faces that are incident to neighbors of $\origsol(\dot{v})$. 
Consequently, any face in $\Gamma$ that has no pistils can safely be ignored by a solution.%

\begin{reduction}\label{reduc:outerp}
Let $(G, \Gamma, \candidate, k)$ be a $\pSSRE$ (\textsc{SSRE}) instance. Any vertex of $G$ not incident to a petal in $\Gamma$ is removed from $G$ and $\Gamma$, alongside all of its incident edges. This results in the reduced instance $(G', \Gamma', \candidate, k)$.
\end{reduction}

\begin{lemma}\label{lma:reduc:outerp}
  Reduction Rule~\ref{reduc:outerp} is sound: given an instance $I=(G, \Gamma, \candidate, k)$ of \textsc{SSRE}, and applying Reduction Rule~\ref{reduc:outerp} to get $I'=(G', \Gamma', \candidate, k)$, it holds that $I$ is a yes-instance if and only if $I'$ is a yes-instance.
\end{lemma}
\begin{proof}

We first show that for a face $f'\in \Gamma'$, but $f' \not \in \Gamma$ there cannot be pistils on its boundary, meaning that any vertex inserted into $f'$ has degree 0 and can be embedded anywhere. Assume there is a pistil $p$ incident to $f'$ and let $F$ be the set of faces incident to $p$ in $\Gamma$. The reduction rule does not remove any vertex incident to a face of $F$ as these faces are all petals. So the faces incident to $p$ in $\Gamma'$ must be the same as those in $\Gamma$. This contradicts that $f'$ is not in $\Gamma$ and hence there cannot be a pistil on $f'$ and a solution for $I'$ is a solution for $I$.

Given a solution $\sigma=(\ssplitsol, \copiessol, \origsol, (N_u)_{u \in \ssplitsol}, \planarreemb^*)$ to $I$, we will show that $\sigma$ is also a solution to $I'$. %
Since every petal of $\Gamma$ is also a face in $\Gamma'$ (we do not remove vertices incident to petals), every embedding in a petal face of $\Gamma$ in $\sigma$ can be replicated in $\Gamma'$. Any copy embedded in a non-petal face of $\Gamma$ must have degree 0 to preserve planarity and can therefore be embedded anywhere.
As a result, $I$ is a yes-instance if and only if $I'$ is a yes-instance, and Reduction Rule~\ref{reduc:outerp} is sound.
\end{proof}

In the remainder, let $\outerp$ be the drawing obtained from $\Gamma$ after exhaustively applying Reduction Rule~\ref{reduc:outerp}. We use \cref{lma:reduc:outerp} every time we apply Reduction Rule~\ref{reduc:outerp} to get the following corollary.

\begin{corollary}\label{cor:reduc:outerp}
Let $I_\text{p}=(G, \Gamma, \candidate, k)$ be a $\pSSRE$ instance, and let $I_\text{o}=(G', \outerp, \candidate, k)$ be the instance obtained from $I_\text{p}$ after applying Reduction Rule~\ref{reduc:outerp} exhaustively. $I_\text{p}$ is a yes-instance if and only if $I_\text{o}$ is a yes-instance.
\end{corollary}

\begin{figure}[b]
	\centering
	\includegraphics{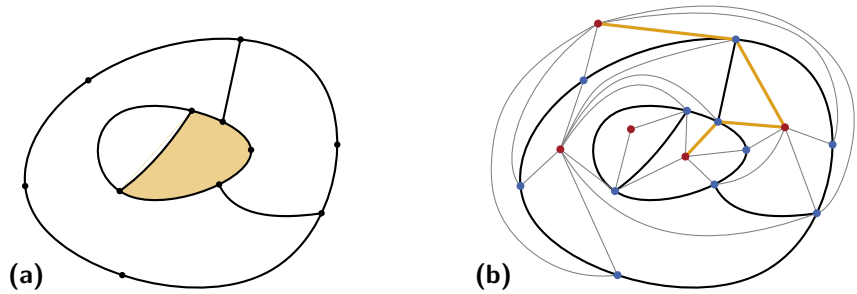}
	\caption{
		\textbf{\textsf{(a)}} A graph and \textbf{\textsf{(b)}} its modified dual. A shortest face path, drawn with yellow edges in \textbf{\textsf{(b)}}, from the yellow face in \textbf{\textsf{(a)}} to the outer-face alternates between face-vertices (in red) and vertex-vertices (in blue).
	}
	\label{fig:path}
\end{figure}%

We introduce the notion of a face path in a drawing $\Gamma$ via a modified bipartite dual graph $D$ of $\Gamma$ whose two sets of vertices are the faces and vertices of $G$ (denoted as \emph{face-vertex} and \emph{vertex-vertex}, respectively) and whose edges link each face-vertex with all its incident vertex-vertices as shown in \cref{fig:path}. Paths in $D$ are called \emph{face paths} and they alternate between faces and vertices of $\Gamma$.
We define the length of a face path as the number of its vertex-vertices. %
Using this notion, a drawing is $k$-outerplanar if all of its vertex-vertices have a face path to the outer face of length at most $k$. We now show that applying the reduction rule exhaustively to an instance of \pSSRE\ transforms drawing $\Gamma$ into a $10k$-outerplanar drawing $\outerp$.

\begin{lemma}\label{outerplma}
If $(G, \Gamma, \candidate, k)$ is a yes-instance, then $\outerp$ $10k$-outerplanar.
\end{lemma}
\begin{proof}
  We show that the shortest face path from any vertex of $\outerp$ to the outer face has length at most $10k$. Let $\ssplitsol$ be the set of copies of $\candidate$ embedded in the drawing $\planarreemb^*$ of a solution.
  We label each face $f$ of $\outerp$ by an arbitrary copy $\dot{v}$ in $\ssplitsol$ that is closest to $f$ (by length of face paths from the faces partitioning $f$ in $\planarreemb^{*}$).
  Note that the length of a face path in $\outerp$ from $f$ to the face of $\outerp$ that $\dot{v}$ is embedded in is at most two since a face is either incident to a pistil, or all its vertices are incident to such a face 
  and all pistils are covered.
  Since $|\ssplitsol| \le 2k$ there are at most $2k$ distinct labels for the faces of $\outerp$.
  Let $v$ be a vertex in $\outerp$ and let $P$ be a shortest face path between $v$ and the outer face.
  We claim that each label appears at most five times in $P$. Assume to the contrary that there are six faces $f_1, f_2, f_3, f_4, f_5, f_6$ with the same label $\dot{u}$ occurring in $P$ in this order. 
  We call $f_{\dot{u}}$ the face in which $\dot{u}$ is embedded in the solution. The faces $f_1$ and $f_6$ must have at most distance 2 from $f_{\dot{u}}$ as otherwise they would be labeled differently. Thus one can find a path from $f_1$ to $f_{\dot{u}}$ of length at most 2, and similarly a path from $f_{\dot{u}}$ to $f_6$ of length at most 2. Thus, there is a shorter path from $f_1$ to $f_6$ which contradicts our claim.
  Hence, $P$ has length at most $10k$, showing that $\outerp$ is $10k$-outerplanar.
\end{proof}

\subsection{Finding a Sphere-Cut Decomposition}\label{sec:ssre-fpt-scbd}

A \emph{branch decomposition} of a (multi-)graph $G$ is a pair $(T,\lambda)$ where $T$ is an unrooted binary %
tree, %
and $\lambda$ is a bijection between the leaves of $T$ and $E(G)$. Every edge $e\in E(T)$ defines a bipartition of $E(G)$ into $A_e$ and $B_e$ corresponding to the leaves in the two connected components of $T-e$.
We define the \emph{middle set $\midset(e)$} of an edge $e \in E(T)$ to be the set of vertices incident to an edge in both sets $A_e$ and $B_e$. The \emph{width} of a branch decomposition is the size of the biggest middle set in that decomposition. The \emph{branchwidth} of $G$ is the minimum width over all branch decompositions of $G$. 

\looseness=-1
A \emph{sphere-cut decomposition} of a planar (multi-)graph $G$ with a planar embedding $\Gamma$ on a sphere $\Sigma$ is a branch decomposition $(T,\lambda)$ of $G$ such that for each edge $e\in E(T)$ there is a \emph{noose} $\eta(e)$: %
a closed curve on $\Sigma$ such that its intersection with $\Gamma$ is exactly the vertex set $\midset(e)$ (i.e., the curve does not intersect any edge of $\Gamma$) and such that the curve visits each face of $\Gamma$ at most once (see Figure~\ref{fig:scd}).
The removal of $e$ from $E(T)$ partitions $T$ into two subtrees $T_1,T_2$ whose leaves correspond respectively to the noose's partition of $\Gamma$ into two embedded subgraphs $G_1,G_2$. Sphere-cut decompositions are introduced by Seymour and Thomas~\cite{SeymourT94}, more details can also be found in~\cite[Section 4.6]{MarxP15}.%
\begin{figure}
	\centering
	\includegraphics{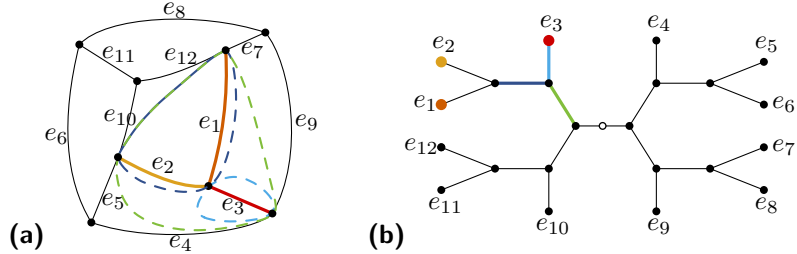}
	\caption{
		\textbf{\textsf{(a)}} A graph and \textbf{\textsf{(b)}} its sphere-cut decomposition. Each labeled leaf corresponds to the same labeled edge of the graph. The middle set of each colored edge in the tree corresponds to the vertices of the corresponding colored dashed noose in the graph.
	}
	\label{fig:scd}
\end{figure}
The \emph{length} of the noose $\eta(e)$ for an edge $e \in E(T)$ is the number of vertices on the noose (or the size of $\midset(e)$) and it is at most the branchwidth %
of the decomposition.
The drawings in this paper are defined in the plane, whereas we need drawings on the sphere for sphere-cut decompositions. However, if we treat the outer face of a planar drawing just as any other face, then spherical and planar drawings are homeomorphic.

An $\ell$-outerplanar graph has branchwidth at most $2 \ell$ \cite{Biedl_2015}. Moreover, each connected bridgeless planar graph of branchwidth at most $b$ has a sphere-cut decomposition of width at most $b$ and this decomposition can be computed in $O(n^3)$ time where $n$ is the number of vertices (this has been shown by Seymour and Thomas~\cite{SeymourT94}, see the discussion by Marx and Pilipczuk~\cite[Section 4.6]{MarxP15}). We transform any bridge in our newly obtained graph into a multi edge to ensure that the graph is bridgeless.

\subsubsection{Obtaining a Bridgeless Graph}\label{sec:ssre-fpt-bridgeless}
While our graph drawing~$\outerp$ is already $\ell$-outerplanar (more specifically $10k$-outerplanar), we have to deal with the \emph{bridges}, i.e., edges of $G$ whose removal disconnect $G$.
There is no guarantee our graph is bridgeless and even if it was required of the input, Reduction Rule~1 might create bridges.
We instead create a new graph $G''$ together with a drawing $\bless$, in which for any bridge $(u,v)$, we add a secondary multi edge between $u$ and $v$ and continue to work with the resulting multigraph.
Note that adding the edge does not affect the outerplanarity and hence does not affect the bound of the width of the decomposition.

\begin{lemma}\label{lma:bridgisouter}
  Given the following two instances of $\pSSRE$, $I_\text{o}=(G', \outerp, \candidate, k)$ and $I_\text{b}=(G'', \bless, \candidate, k)$, which differ only in the graph $G'$ and $G''$ and their respective drawings $\outerp$ and $\bless$, where $G''$ is a copy of $G'$ plus a duplicate of every bridge, then $I_\text{o}$ is a yes-instance if and only if $I_\text{b}$ is a yes-instance.
\end{lemma}

\begin{proof}
Given a solution $\sigma_\text{o}=(\ssplitsol, \copiessol, \origsol, (N_{\dot{v}})_{\dot{v} \in \ssplitsol}, \planarreemb^*)$ to $I_\text{o}$, we will show that $\sigma_\text{o}$ is also a solution to $I_\text{b}$.
We can extend $\sigma_\text{o}$ to solve $I_\text{b}$ as follows. All the faces of $\outerp$ exist in $\bless$ and have the same set of vertices incident to them. This means that any split vertex $s$ inserted into a face $f$ of $\outerp$ to cover a set of vertices $N(s)$ can be inserted in the same face of $\bless$ and cover the same set of vertices. 
So $\sigma_\text{o}$ is a solution to $I_\text{b}$.

Given a solution $\sigma_\text{b}=(\ssplitsol, \copiessol, \origsol, (N_{\dot{v}})_{\dot{v} \in \ssplitsol}, \planarreemb^*)$ to $I_\text{b}$, we can extend $\sigma_\text{b}$ to solve $I_\text{o}$ in the same way as for the previous paragraph with one exception. There may be split vertices that are embedded into a face bounded by two multi-edges $(u,v)$ of $\bless$ that does not exist in $\outerp$. Such a split vertex~$\dot{s}$ would have one or both of $u,v$ as neighbors. If the split vertex only has one neighbor then it can be embedded in any face incident to that neighbor without disturbing anything. If $N(\dot{s})=\{u,v\}$, then there is only a single face of $\outerp$ in which they both lie (by virtue of being a bridge).
If there are no vertices embedded on that face and we can freely put $\dot{s}$ in it and have it reach its neighborhood without risk of creating a crossing. Otherwise, if there are other copies of split vertices embedded in that face, we can still embed $\dot{s}$ sufficiently close to edge $(u,v)$ and connect it to $u$ and $v$ by two crossing-free edges. 
This means that after embedding all such vertices, we have found a solution to $I_\text{o}$, and with this we showed that $I_\text{o}$ is a yes-instance if and only of $I_\text{b}$ is a yes-instance.
\end{proof}
}

\newcommand\pREDP{\textsc{Split Set Re-Embedding Decomp}}

\toappendix{
As a result, we have a bridgeless graph with a $10k$-outerplanar drawing which thus admits a sphere-cut decomposition of width bounded by $20k$, as discussed before. Moreover, we can compute this decomposition in $O(n^{3})$ time \cite{Biedl_2015,MarxP15,SeymourT94}.

\subsection{Initializing the Dynamic Programming}\label{sec:ssre-fpt-dp}

In the previous sections, we used branching and computed a set~$\ssplit$ as well as the mappings~$\copies$ and~$\orig$ for each branch. We now use dynamic programming on the sphere-cut decomposition $(T_0,\lambda)$ of the bridgeless $10k$-outerplanar graph $G'$, obtained previously, and its drawing $\Gamma'$ to find the remaining elements of the solution: $(N_{\dot{v}}^*)_{\dot{v} \in \ssplit}$ and $\planarreemb^*$.
We have therefore effectively reduced \pSSRE\ to the following more restricted problem:

\begin{problem}[\pREDP]
  Let the following be given: a graph $G=(V,E)$, a set $\candidate\subseteq V$, an integer $k \ge |\candidate|$, a drawing $\Gamma$ on the sphere of the $10k$-outerplanar bridgeless graph $G' = (V', E') := G[V \setminus \candidate]$, its sphere-cut decomposition $(T_0,\lambda)$, and, as guessed by the initial branching, two mappings $\orig$ and $\copies$ and a graph $G_\ssplit$ on the set of vertices $\ssplit$.
  The task is to decide whether there is a solution $(\ssplitsol, \copiessol, \origsol, (N_{\dot{v}}^*)_{\dot{v} \in \ssplitsol}, \planarreemb^*)$ to the instance $(G, \candidate, \Gamma, k)$ of \pSSRE\ 
  that coincides with the guessed branch, i.e., $\ssplitsol = \ssplit$, $\origsol = \orig$, $\copiessol = \copies$, and $G^*[\ssplitsol] = G_\ssplit$, where $G^*$ is the corresponding solution graph.
\end{problem}

}

\toappendix{
We first transform $T_0$ into a rooted tree $T$ by choosing an arbitrary edge $e_r=(r_1,r_2)\in T_0$ and subdividing it with a new root vertex~$r$. %
This induces parent-child relationships between all the vertices in $T$ and we set for a given vertex $t\in V(T)$ with parent $p$, the noose $\eta(t) = \eta((p,t))$ corresponding to the set $\midset((p,t))$.
Since for each $t \in V(T) \setminus \{r\}$ the parent~$p$ is unique, we simply use $\midset(t)$ instead of $\midset((p,t))$.
Additionally we put $\midset(r_1)=\midset(r_2)=\midset(e_r)$. %

The dynamic program works bottom-up in $T$, considering iteratively larger subgraphs of the input graph~$G'$. It determines how partial solutions look like on the interface between subgraphs and the rest of $G'$.
We need the following notions to define this interface.

  For a vertex $t\in V(T)$ and the noose $\eta(t)$ associated to it, we define the subgraph $G'_t$ of $G'$ as the subgraph obtained from the union of the edges that correspond to the leaves in the subtree of $T$ rooted at $t$.
  We say that a subgraph of $G'$ is \emph{inside} the noose $\eta(t)$, %
  if it is a subgraph of $G'_t$. %
  For a noose $\eta(t)$ of $(T, \lambda)$, the \emph{processed} faces are faces of the subgraph $G'_t$ that have all their incident vertices inside or on $\eta(t)$; the \emph{current} faces of $\eta(t)$ are the faces of $G'$ that have vertices inside, on, and outside of the noose $\eta(t)$.

Nooses in sphere-cut decompositions by definition pass through a face at most once, so it is not possible for a face to have an intersection with a noose that contains more than two vertices.
Additionally, since $\Gamma$ is embedded on the sphere, there is no outer face.%

Next, we define partial solutions for the subgraphs of $G'$. Intuitively, a partial solution for a subgraph $G'_{t}$ is a planar drawing~$\Gamma'$ of $G'_{t}$ using a set~$\ssplitp$ of copies that covers all of the pistils inside the noose~$\eta(t)$ defining $G'_{t}$.

\begin{definition}
    A \emph{partial solution} for $G'_t$ is a tuple $(\ssplitp, (N_{\dot{v}})_{\dot{v} \in \ssplitp}, \Gamma')$, s.t.\ $\ssplitp \subseteq \ssplit$ and:
\begin{enumerate}[(i)]
\item for every $\dot{v} \in \ssplitp$, we have $N_{\dot{v}} \subseteq N_G(\orig(\dot{v}))$,
\item planar drawing~$\Gamma'$ extends $\Gamma$ by embedding each copy $\dot{v} \in \ssplitp$ in a face of $G'_t$,
\item for every $v \in \orig(\ssplitp)$, $\{N_{\dot{v}} \mid \dot{v} \in \copies(v)\}$ is a partition of some subset of $N_G(v)$,
\item for each pistil $p$ of $G'$ that is inside (and not on) the noose $\eta(t)$ and for each neighbor $v \in N_G(p) \cap \candidate$ there is a copy $\dot{v} \in \copies(v)$ such that $p \in N_{\dot{v}}$ and $\dot{v} \in N_{\Gamma'}(p)$, and\label{cond:pistils-covered}
\item $G(\Gamma')[\ssplitp] = G_\ssplit[\ssplitp]$, where $G(\Gamma')$ is the graph corresponding to the partial solution drawing $\Gamma'$ and $G_\ssplit$ is the graph guessed by branching.%
\end{enumerate}
\end{definition}

We now describe the information that is required during dynamic programming to find partial solutions. To model a partial solution in a current face, we describe the combinatorial embedding of the split vertices embedded inside that face using a structure called a nesting graph.
Let $\eta$ be a noose and let $f$ be a face that is current for~$\eta$.
A \emph{nesting graph} $C_f$ for $f$ and a vertex set~$S_f\subseteq\ssplitp$ to be embedded in $f$ is a combinatorially embedded graph with vertex set $S_f \cup C$, such that $C$ induces a cycle that encloses all vertices in $S_f$. It is a combinatorial representation of the boundary of $f$ and its vertices represent a subset of the vertices of the boundary of $f$ (possibly merged).

We call the vertices in $C$ \emph{cycle vertices}.
If $S_f=\emptyset$, the nesting graph is the empty graph. If $S_f=\{s\}$, the cycle is of length 2 and both cycle vertices share an edge with $s$. If $|S_f|>1$, the following conditions must be satisfied:
\begin{enumerate}[(i)]\label{nestinggraph}
\item embedded graph~$C_f$ is planar and all vertices in $C$ must be incident with its outer face,
\item the vertices of 
$\cup_{\dot{v}\in S_f}N(\dot{v})\setminus S_f$ lie in a cycle $C$ and $S_f$ is embedded inside of $C$,%
\item each vertex $u\in C$ has exactly one neighbor in $S_f$,\label{cond:nest-graph-cycle-degree}
\item for any two vertices $c_1,c_2\in C$ that have the same neighbor $\dot{s}\in S_f$, it holds that $c_1$ and $c_2$ are not neighbors in $C_f$, and lastly\label{cond:nest-graph-cycle-nonneighbors}
\item\label{cond:nest-graph-induced} $C_{f}[S_f] = G_{\ssplit}[S_f]$.

\end{enumerate}

Intuitively, to check which pistils in the current face $f$ can be covered by $S_f$ while using the embedding of $C_f$, we can imagine $C_f$ embedded inside of $f$ and attempt to draw crossing-free edges from the vertices of $C$ to the pistils of $f$.

To describe the combinatorial embeddings of the vertices inside of the cycle of nesting graphs, we introduce the notion of \emph{compatibility} (see Figure~\ref{fig:nest}).

\begin{definition}
  Given a face $f$ and two copies $\dot{u},\dot{v} \in \ssplit$ to be embedded in $f$ with their respective neighborhoods $N(\dot{u})$ and $N(\dot{v})$ incident to $f$, 
  we say that $\dot{v}$ \emph{is compatible with} $\dot{u}$ in $f$ if in the cyclic ordering of $N(\dot{u})$ and $N(\dot{v})$ around $f$, the respective neighborhoods do not interleave.
\end{definition}

\begin{figure}
	\centering
	\includegraphics[width=\textwidth]{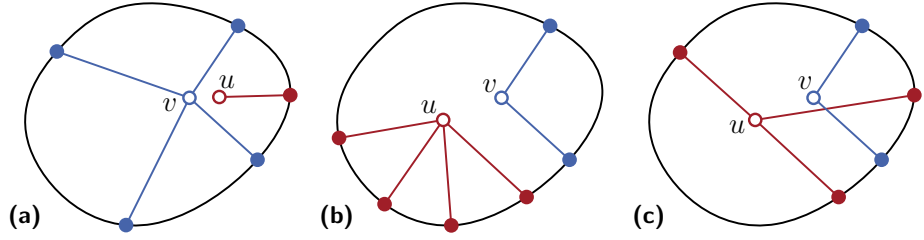}
	\caption{
		\textbf{(a)} A degree-one vertex $u$ is compatible with any other vertex, here only $v$. \textbf{(b)} Vertex $u$ is still compatible with $v$. \textbf{(c)} Vertices $u$ and $v$ are not compatible.
	}
	\label{fig:nest} 
\end{figure}

\looseness=-1
Since current faces of $G'_t$ are not fully inside~$\eta(t)$, we want to specify which part of a nesting graph $C_f$ is used in a current face $f$ to cover incident pistils. To achieve this we keep track of only the first and last vertices $\ps,\pe$ on the cycle $C$ of $C_f$ that connect to pistils in~$f$. %
That is, a clockwise traversal of $C$ from $\ps$ to $\pe$ visits all vertices used to cover pistils of $G'_t$. If only a single pistil is covered, then $\ps =\pe$;  if no pistil is covered $\ps$ and $\pe$ can be undefined.

While all pistils inside the noose~$\eta(t)$ of $G'_t$ are covered in a partial solution, the vertices on $\eta(t)$ can have missing neighbors.
That is, for a pistil $p$ on $\eta(t)$ with neighborhood $N_G(p)$ in $G$ and $N_{\Gamma'}(p)$ in $\Gamma'$, 
and $v\in\candidate\cap N_G(p)$ 
such that $\copies(v)\cap N_{\Gamma'}(p)=\emptyset$, vertex~$v$ is a \emph{missing neighbor} of~$p$. We define $X_t(p)$ as the set of missing neighbors of $p$.

Using the elements described above, we build a \emph{signature} for a node $t\in T$ modeling a possible solution.
\begin{definition}
  A \emph{signature} on $t \in V(T)$ is a tuple $(\inner(t), \cycle(t), M_t, \noose(t))$
  where:
  \begin{enumerate}
  \item $\inner(t)\subseteq \ssplit$,
  \item $\cycle(t)$ is a set containing a nesting graph $C_{f}$ for each current face $f$ of $G'_t$ such that no two of the graphs share a vertex in $\ssplit$,
  \item $M_t \colon \cycle(t) \to V(\cycle(t))\times V(\cycle(t))$ maps each nesting graph to a pair $(\ps,\pe)$, and
  \item $\noose(t)$ is a list of $|\midset(t)|$ sets of missing neighbors, defined as $\langle X_t(p) \rangle_{p \in \midset(t)}$.
  \end{enumerate}
\end{definition}
For a given signature on $t \in V(T)$, the set $\inner(t)\subseteq \ssplit$ tracks the split vertices embedded in the processed faces of~$G'_t$, graph~$\cycle(t)$ and $M_t$ track the embedding information in current faces, and $\noose(t)$ tracks the missing neighbors of the noose vertices~\cref{fig:sigcontent}.

\looseness=-1
We now show that every partial solution has a corresponding signature. %
For a partial solution $(\ssplitp, (N_{\dot{v}})_{\dot{v} \in \ssplitp}, \Gamma')$ for $G'_t$ we construct a signature tuple $(\inner(t),\cycle(t), M_t, \noose(t))$ (or $(\inner,\cycle, M, \noose)$ when $t\in V(T)$ is clear from context) as follows.
The set $\inner$ is composed of all vertices of $\ssplitp$ embedded in processed faces of $G'_t$.
To construct $\noose$, we create $X_t(p) \in \noose$ for each noose vertex~$p$, and add to $X_t(p)$ all neighbors in $\candidate$ for which no copy is adjacent to $p$ in~$\Gamma'$: $X_t(p) = (N_G(p) \cap \candidate) \setminus \orig(N_{\Gamma'}(p) \cap \ssplitp)$.
Lastly, for a current face $f$ of $G'_t$ where $S_f\subseteq\ssplitp$ are the copies embedded in $f$ in $\Gamma'$, we find the nesting graph~$C_f$ by transforming $f$ and the graph induced by $S_f$:
If $|S_f|=1$, graph $C_{f}$ will consist of that vertex embedded inside a 2-cycle, with outgoing edges from the split vertex to both cycle vertices.
If $|S_f|\geq 2$, the construction is as follows (see Figure~\ref{fig:nesting}).
\begin{itemize}
    \item For each $v\in V(G'_t)$ such that $N(v)\cap\ssplitp=\emptyset$ we contract one of its outgoing edges in~$f$. We repeat this process until all remaining vertices of $f$ are incident with a vertex of~$\ssplitp$.
    \item For each $u\in V(G'_t)$ we let $\dot{s}_1,\dots,\dot{s}_i$ be the ordering of $N(u)\cap \ssplitp$ clockwise around~$u$. We replace $u$ by the path $(u_1,u_2),\dots,(u_{i-1},u_i)$ on which we then add the edges $(\dot{s}_1,u_1),\dots,(\dot{s}_i,u_i)$. Thus each vertex on $f$ is adjacent to at most one vertex of $\ssplitp$.
    \item For $e=(u_1,u_2)$ incident to $f$, if there is a split vertex $\dot{s}\in \ssplitp$ such that $(u_1,\dot{s}),(u_2,\dot{s})$ in $\Gamma'$, we contract $e$ such that the new vertex obtained has a single edge to~$\dot{s}$.
    \end{itemize}
Since $G'_t$ is bridgeless and we initialized $C_f$ as $f$, $C_f\setminus S_f$ is necessarily a cycle~$C$.
To obtain $M$, for each graph in $C_f$ we find the section of $C$ used for the coverage of~$G'_t$, determine the first and last vertex in clockwise ordering on cycle~$C$, here $\ps$ and $\pe$, and set $M_t(C_f) = (\ps, \pe)$.
The obtained signature is the \emph{signature of the partial solution} and each such signature is \emph{valid}.

During the dynamic programming we enumerate all signatures, and hence we first determine how many distinct signatures can exist. To bound the number of different signatures, we need an upper bound on the number of vertices in a nesting graph.
\begin{lemma}\label{lma:cyclevertcount}
A nesting graph with $s$ split vertices embedded inside the cycle $C$ has at most $2s$ vertices on the cycle.
\end{lemma}

\begin{proof}
  Let  $C_f=(V,E)$ be  a nesting graph, where $V=C \cup S_f$ with $C$ being the cycle defining the outer face and $S_f$ the vertex set embedded inside the cycle and let $|S_f|=s$.
  We remove all edges $e=(u,v)\in E$ such that $u,v\in S_f$ and all vertices in $S_f$ that have degree zero after the removal.
  We obtain a new set of faces $F$.
    Euler's formula tells us that for this modified nesting graph $G' = (V',E')$ we have $|V'|-|E'|+|F|=2$. 
    We first notice $|V'|\leq s+|C|$ and thus
    \begin{align}\label{eq:euler1}
      s + |C| - |E'| + |F| & \geq 2.
    \end{align}
    Considering that $E'=\{(u,v), u\in C, v\in C\}\cup \{(u,v), u\in C, v\in S_f\}$, that by property~(i) of nesting graphs $|\{(u,v), u\in C, v\in C\}|=|C|$, and that $|\{(u,v), u\in C, v\in S_f\}| = |C|$ (from (ii)), we obtain $|E'|=2|C|$.
    Plugging this into \cref{eq:euler1}, we have
    \begin{align*}
      &s + |F|  \geq 2 + |C|\\
      &s + |F| - 2  \geq |C|.
    \end{align*}
    We lastly show that $|F|\leq \frac{|C|}{2}$.
    Assume for a contradiction that a face $f\in F$ that is not the outer face has only one edge $e=(u,v)$ of $C$ on its boundary.
    We traverse $f$ starting from $v$, visiting its neighbor~$v_s$ different from $u$.
    Note that $v_{s} \in S_f$.
    Then as we have removed all edges incident to $v_s$ having their other endpoint in~$S_f$, the next vertex $v_c$ is necessarily on the cycle~$C$.
    Considering property \eqref{cond:nest-graph-cycle-nonneighbors} of nesting graphs we find that $v_c\ne u$.
    Thus our traversal continues and property \eqref{cond:nest-graph-cycle-degree} implies that the next vertex must be a vertex on the cycle different from $u$ and $v$, which contradicts our assumption that an inner face of $C_f$ only has a single cycle edge incident to itself.
    
    Hence, $|C|\leq s-2+\frac{|C|}{2}$, giving us $|C|\leq 2s-4\leq 2s$. %
\end{proof}

Next we bound the number of distinct signatures for $k$ split operations.
}

\begin{restatable}{lemma}{countingsig}
  \label{lma:sigcount}
  The number $N_s(k)$ of possible signatures is upper bounded by $2^{O(k^2)}$.
\end{restatable}

\appendixproofwithrestatable{\countingsig*}
{
  \begin{proof}
  We first count the number of all possible sets $\inner$.
  These are necessarily subsets of $\ssplit$.
  Thus, there are $2^{2k}$ such sets. 

  To find the number of the possible sets $\noose$, we compute the set of all candidate vertices that are neighbors to the noose vertices.
  As our input graph must be a $10k$-outerplanar graph and we know that those have sphere-cut decompositions of branchwidth upper bounded by $20k$, there are at most $20k$ vertices on a noose in the decomposition.
Since a single vertex has at most $k$ split neighbors,
the set is of size at most $20k^2$. This means that there are at most $2^{20k^2}$ sets $\noose$.

To find the number of all the possible sets $\cycle$ of nesting graphs, consider constructing an auxiliary graph $\hat{C}$ that consists of the union of all nesting graphs in $\cycle$ connected by edges in an arbitrary tree-like fashion.
To find an upper bound for the number of sets $\cycle$ it is enough to find an upper bound on the number of (combinatorially) embedded graphs $\hat{C}$.
By \cref{lma:cyclevertcount} and since no two nesting graphs in $\cycle$ share a vertex in $\ssplit$ by definition of signatures it follows that $|V(\hat{C})| \leq 4k$ (recall that $|\ssplit| \leq 2k$).
Consider taking~$\hat{C}$ and triangulating it; in this way it becomes triconnected.
Thus $\hat{C}$ is a partial drawing of a triconnected embedded planar graph with at most $4k$ vertices.
The number of partial drawings of a fixed triconnected embedded planar graph with at most $4k$ vertices is at most $2^{O(k)}$ because such a graph contains $O(k)$ edges.
It remains to find an upper bound on the number of triconnected combinatorially embedded planar graphs.
Since each triconnected planar graph with $4k$ vertices has $O(k)$ different combinatorial embeddings (the embedding being fixed up to the choice of the outer face), an upper bound is $O(k)$ times the number of planar graphs with $4k$ vertices.
To bound the number of planar graphs with $4k$ vertices, recall that each planar graph is $5$-degenerate, that is, it admits a vertex ordering such that each vertex~$v$ has at most $5$ neighbors before $v$ in the ordering.
Thus, there are at most $((4k)^5)^{4k}$ planar graphs with $4k$ vertices because all planar graphs on a given vertex set can be constructed by, for each vertex~$v$, trying all possibilities to select the at most five neighbors of $v$ occurring before $v$ in the degeneracy ordering.
Overall, we obtain an upper bound of $2^{O(k)} \cdot O(k) \cdot O(k)^{O(k)} = 2^{O(k \log k)}$ for the number of sets $\cycle$.
 
We lastly need to find all mappings $M_t$ for the nesting graph set.
The trees have at most $4k$ vertices, which means in the worst case a cycle has $8k$ vertices.
This gives us at most $\binom{8k}{2}$ vertex pairs per nesting graph (as well as $\emptyset$) as possible outputs for every nesting graphs.

This means that there are at most $N_s(k)= 2^{2k} \cdot 2^{20k^{2}} \cdot 2^{O(k \log k)} \cdot \binom{8k}{2} = 2^{O(k^{2})}$ possible signatures for a given noose.
  \end{proof}
}

\looseness=-1
\subsubsection{Dynamic programming.}
Finally, we give an overview of how the valid signatures are found. In each branch, we perform bottom-up dynamic programming on $T$.
We want to find a valid signature at the root node of $T$, and we start from the leaves of $T$. 
Each leaf corresponds to an edge $(u_1,u_2)$ of the input graph~$G$, for which we consider all enumerated signatures and check if a signature is valid and thus corresponds to a partial solution. Such a partial solution should cover all missing neighbors of $u_1$ and $u_2$ not in $\noose=\{X(u_1),X(u_2)\}$, using for each incident face $f$ the subgraph of $C_f\in\cycle$ as specified by $M$\ifArxiv ~(see Appendix~\ref{subsec:dp-leaf}). \else. \fi

For internal nodes of $T$ we merge some pairs of valid child signatures corresponding to two nooses $\eta_1$ and $\eta_2$.
We merge if the partial solutions corresponding to the child signatures can together form a partial solution for the union of the graphs inside $\eta_1$ and $\eta_2$. The signature of this merged partial solution is hence valid for the internal node when
(1) faces not shared between the nooses do not have copies in common, (2) shared faces use identical nesting graphs and (3) use disjoint subgraphs of those nesting graphs to cover pistils, and (4) %
noose vertices have exactly a prescribed set of missing neighbors\ifArxiv ~(details in Appendix~\ref{subsec:dp-intern}). \else. \fi
Thus we can find valid signatures for all nodes of~$T$ and notably for its root. 
If we find a valid signature for the root, a partial solution~$(\ssplitp, (N_{\dot{v}})_{\dot{v} \in \ssplitp}, \Gamma')$ must exist. In~$\Gamma'$ all pistils are covered and it is planar, as the nesting graphs are planar and they represent a combinatorial embedding of copies that together cover all pistils. It is possible that certain split vertices are in no nesting graph, and hence $\ssplit\setminus\ssplitp \neq \emptyset$. We verify that the remaining copies that are pistils in $\ssplit\setminus\ssplitp$ induce a planar graph, which allows us to embed them in a face of $\Gamma'$ to obtain the final drawing. The running time for every node of~$T$ is polynomial in $N_s(k)$, thus, over all created branches \pSSRE\ is solved in $2^{O(k^2)} \cdot n^{O(1)}$ time\ifArxiv ~(see Appendix~\ref{subsec:dp-root}). \else. \fi
\begin{restatable}{theorem}{fptpssre}
  \label{thm:ssre-fpt}
  \pSSRE\ can be solved in $2^{O(k^2)} \cdot n^{O(1)}$ time, using at most $k$ splits on a topological drawing~$\Gamma$ of input graph~$G$ with $n$ vertices.
\end{restatable}

\toappendix{
\subsection{Finding Valid Signatures for Leaf Nodes}\label{subsec:dp-leaf}

The next three sections will each explain part of our dynamic programming approach, starting with the leaf nodes of the sphere-cut decomposition tree $T$. By definition the leaf nodes of $T$ form a bijection with the edges of $\Gamma$.
The first step of the dynamic programming algorithm is to find the set of valid signatures on the leaf nodes of $T$.
To do this, we loop over each possibles signature for each leaf $t\in V(T)$, and check whether it is valid.
In the following, we show how we determine whether a signature for a leaf~$t\in V(T)$ is a valid signature for~$t$, that is, the signature corresponds to a partial solution for the subinstance induced by the edge of $\Gamma$ that corresponds to~$t$.
All valid signatures found for $t$ are added to $D(t)$, the table of valid signatures on the node~$t$. 
The signatures we created by enumerating all possibilities are not all useful, a signature's nesting graph might not be an actual nesting graph so we check beforehand that they satisfy the necessary properties.
From the non-discarded signatures we can begin working on the decomposition.
In the following, we refer to the cycle and its incident vertices as the cycle of the nesting graph.

For a leaf node~$t\in T$ with corresponding edge $e=(u,v)$ in $\Gamma$ and given the signature for $t$ $\sig=(\inner(t), \cycle(t)=\{C_1,C_2\}, M_t, \noose(t)=\{X_t(u),X_t(v)\})$, %
the algorithm will proceed in two steps to cover $u$ and $v$ by embedding vertices in the two faces of $\Gamma$ incident to~$e$. 
It will first check elements of the signature and whether or not there are pistils on the noose to assess trivially valid and invalid signatures.
Then it will execute a routine based on branching and traversing the cycle of the nesting graph to assess the remaining signatures.
We first define the set $N_u(t)=N_G(u)\setminus X_t(u)$ %
containing the neighbors of $u$ that should cover $u$, as they are not in the missing neighbors set %
in $G'_t$ and analogously $N_v(t)$. %
The trivial checks are the following:
\begin{enumerate}
\item If $C\in\cycle$ is not a nesting graph, the signature is not valid.
\item If $\inner\ne\emptyset$ then the signature is not valid -- as $G'_t$ is only an edge, there is no processed face.
\item If $N_u(t)=N_v(t)=\emptyset$, then $\sig\in D(t)$: there is no vertex to cover in the subgraph, hence the signature is valid (even if $M_t\ne\emptyset$ the vertices in the nesting graph will just not be used).
\item If $N_u(t),N_v(t)\ne\emptyset$, but $\cycle=\{\emptyset,\emptyset\}$ or $M_t=\emptyset$, then there is no available vertex to cover the pistils, hence the signature is not valid.
\end{enumerate}

For the general case, we now introduce algorithm $\algA$ to check the validity of the signatures on leaf nodes, using the following notation.
Let $\ps^1,\pe^1 \in C_1$ and $\ps^2,\pe^2 \in C_2$ such that $M_t(C_1)=(\ps^1,\pe^1)$ and $M_t(C_2)=(\ps^2,\pe^2)$.
A clockwise traversal of the face~$f_1$ incident to $e$ in $\Gamma$ associated with $C_1$ encounters the endpoints of $e$ one after the other in a specific order: w.l.o.g.\ we assume here that $u$ is visited right before $v$, meaning for the other face $f_2$ incident to $e$, a clockwise traversal finds $v$ right before $u$. 

In its execution, algorithm $\algA$ first branches over all possible partitions of $N_u(t)\cup N_v(t)$ into two sets $S_1,S_2$.
Then, for each branch, it creates two new branches, assigning (1) $S_1$ to $C_1$ and $S_2$ to $C_2$ and (2) vice versa.

  Next, algorithm $\algA$ creates a branch for each ordering of $N_u(t)\cap S_1$.
  In each branch, the algorithm proceeds as follows.
  It creates a queue $Q_1$ from the ordering of $N_u(t)\cap S_1$ and a queue $Q_2$ containing the elements of the cycle of $C_1$ in the order of a clockwise traversal starting with~$\ps^1$.
  It then checks whether the first element $q_1$ of $Q_1$ is equal to the first element $q_2$ of $Q_2$.
  If so, $q_1$ is removed from $Q_1$.
  Intuitively, this corresponds to covering the pistil~$u$ by a copy of $q_1$ connected to cycle vertex~$q_2$.
  Otherwise, the cycle vertex~$q_2$ is not connected to a copy of the current missing neighbor $q_1$, thus $q_2$ is removed from $Q_{2}$.
  It repeats the checks and removals of the first elements until either (1) $Q_1$ becomes empty or (2) $Q_2$ becomes empty.
  In the second case, the algorithms discards the current branch.
  In the first case, the algorithms proceeds in the same way with $N_v(t) \cap S_1$ and $C_2$.
  If in at least one branch for an ordering of $N_u(t)\cap S_1$, the algorithm never reaches case (2), the algorithm accepts and puts $\sig \in D(t)$.
However, if all branches are discarded then the signature is not valid. %
Thus, we fill $D(t)$ with all signatures found to be valid.

\begin{lemma}\label{lma:algoA}
  Algorithm $\algA$ is correct.
\end{lemma}
\begin{proof}
  For a leaf $t\in T$ such that $\lambda(t)=e$ where $e=(u,v)$ in $E(G')$, we show that for $\sig=(\inner(t), \cycle(t), M_t, \noose(t))$ algorithm $\algA$ finds $\sig\in D(t)$ if and only if the input signature is valid, meaning we can find a partial solution $\sigma=(\ssplitp, (N_{\dot{v}})_{\dot{v} \in \ssplitp}, \Gamma')$ from $\sig$. In the following proof, for ease  of reading, we identify the vertices on the cycle of a nesting graph with corresponding split neighbor in that graph.

Given a partial solution~$\sigma$ for the signature $\sig$, and embedded subgraph graph $G'_t=(\{u,v\},e)$,
we show that algorithm $\algA$ sets $\sig\in D(t)$. Let $f_1$ and $f_2$ be the faces incident to $e$ and let $C_1$ and $C_2$, respectively, be their assigned nesting graphs. By definition of the signature of a partial solution, since $G'_t$ has no processed faces, necessarily $\inner(t)=\emptyset$. Additionally, if $N_u(t)=N_v(t)=\emptyset$, necessarily $\cycle(t)=\{\emptyset,\emptyset\}$ and $M_t=\noose(t)=\emptyset$. Similarly if $N_u(t),N_v(t)\ne\emptyset$, then $\cycle(t),M_t,\noose(t)\ne\emptyset$. 
Lastly all edges between split vertices in the solution are fully in $f_1$ or $f_2$, therefore the preliminary checks do not discard the signature.
If $N_u(t),N_v(t)\ne\emptyset$ in the partial solution, $u$ and $v$ have a subset of their neighborhood in $f_1$ and/or $f_2$. Algorithm $\algA$ attempts all partitions of $N_u(t)\cup N_v(t)$ into two sets, hence two branches will necessarily correspond to the partition of the solution. Additionally, since both options of assigning a set of the partition to a face are considered, one branch will assign the correct subsets of $\ssplitp$ to the correct faces. We place ourselves in the search tree nodes corresponding to this path.
W.l.o.g., by definition of the signature of a partial solution, the clockwise traversal of the neighborhood of $u$ and $v$ in $f_1$ corresponds to a sub-sequence of a clockwise traversal of the vertices on the cycle of $C_1$ between $\ps^1$ and $\pe^1$.
This sub-sequence of cycle $C_1$ will correspond to one of the orders we branch over: since a solution exists, this order is one of the possible orders of $N_u(t)\cap S_1$, $N_v(t)\cap S_1$, hence, each search will always find a copy of the neighbor it is trying to cover the pistil with, meaning $\algA$ finds the signature of a partial solution to be valid.

Given $\sig$ such that $\algA$ finds $\sig\in D(t)$, we show that $\sig$ corresponds to a partial solution~$\sigma=(\ssplitp, (N_{\dot{v}})_{\dot{v} \in \ssplitp}, \Gamma')$. We set $V(C_1)\cup V(C_2)=\ssplitp$. If $N_u(t)=N_v(t)=\emptyset$, $\sig\in D(t)$: if no endpoints of the edge are pistils, $\ssplitp=\emptyset$, there are no neighborhoods to cover, and the drawing of the input edge is necessarily planar. Hence, if $N_u(t)=N_v(t)=\emptyset$, then $\algA$ is correct.
Otherwise, $\algA$ returns $\sig\in D(t)$ when a branch of the search tree encounters all the neighbors it searches for before completing the nesting graph cycle traversal in both faces. In this case, we construct the partial solution in the following way. 
We place the nesting graphs in their corresponding face. We follow the path of the search tree that completes the search. When encountering a search tree node, we draw an edge between the cycle vertex found by the traversal and the pistil it is in the neighborhood of. Because both traversals ($(u,v)$ and the cycle) are in the same direction, the drawing is planar. Indeed there could only be a crossing between an edge that has $u$ as an endpoint and an edge that has $v$ as an endpoint, but since their neighborhoods are successive (they do not alternate), this is impossible. We can then remove the cycle edges and contract the edges between the pistils and the cycle vertices. The nesting graph is planar and these operations preserve planarity ensuring the final drawing is planar. 
The set of vertices $S_f$ inside the cycle of both nesting graphs $C_1,C_2$ is such that for every $s\in\candidate$ that has a copy in $S_f$,
we have that $\{N_{\Gamma'}(\dot{s}) \mid \dot{s} \in \copies(s)\cap\ S_f\}$ 
is a partition of some subset of $N_G(s)$, more specifically, a subset of $\{u,v\}$. %
There are no pistils inside of the noose, $\ssplitp\subseteq \ssplit$ and for every $\dot{s} \in \ssplitp$ we have $N_{\dot{s}} \subseteq N_G(\orig(\dot{s}))$, showing that $\sig$ is a partial solution.%

Thus, $\algA$ is correct.%
\end{proof}

\begin{lemma}\label{lma:leafruntime}
 For $t\in V(T)$ a leaf node, we fill $D(t)$ in $\mathcal{O}(N_s(k)k2^{2k}2k!)$ time.%
\end{lemma}

\begin{proof}
The trivial checks can be done in constant time. For the following step, there are at most $k$ missing neighbors to $u$ and $v$ each, meaning $2^{2k-1}$ partitions of the missing neighbors in the subset, giving $2^{2k}$ branches as each set creates one branch for each face assignment. Each of these branches will further branch into all the possible orders of their vertex set, meaning a single branch will further create at most $2k!$ branches.
Each branch traverses the two sub-cycles once. Since any cycle has at most $4k$ vertices this takes $\mathcal{O}(k)$ time. 
Overall for a given signature it takes $\mathcal{O}(k2^{2k}2k!)$ time to check if it is valid, and it takes $\mathcal{O}(N_s(k) k2^{2k}2k!)$ time to fill the table for a leaf node.
\end{proof}

\subsection{Finding Valid Signatures for Internal Nodes}\label{subsec:dp-intern}
In this section, we describe the algorithm that fills the table of valid signatures for internal nodes of $T$. For $t,c_1,c_2\in V(T)$ such that $c_1,c_2$ are the children of $t$ and given tables $D(c_1)$ and $D(c_2)$ of valid signatures for $c_1, c_2$ respectively, we will describe an algorithm that checks for each signature pair, whether there exists a corresponding valid signature for $t$.

Given the signatures $s_1,s_2$ of two children $c_1,c_2$ of the same node $t$, we introduce algorithm \algB\ that tests the two signatures for the following four properties. If the pair $s_1,s_2$ has all the required properties, \algB\ finds the signatures to be \emph{consistent} and combines them to compute a signature $s_t = \algB(s_1, s_2)$ for the parent node. 
The tests are the following.%

\begin{enumerate}
    \item Firstly, $\inner(c_1)\cap \inner(c_2)= \emptyset$, $\inner(c_1)\cap \souter(c_2)= \emptyset$, and $\souter(c_1)\cap \inner(c_2)=\emptyset$, where $\souter$ are the vertices of the graphs of $\cycle$. This ensures vertices shared by both signatures can only belong to nesting graphs.%
    \item We then check that $\souter(c_1)\cap \souter(c_2)$ concerns only vertices in nesting graphs corresponding to faces that intersect both noose. Thus for all current faces $f\in \eta(t)$ current in $\eta(c_1)$ and $\eta(c_2)$ with $C_f(c_1)\in \cycle(c_1)\in s_1,C_f(c_2)\in \cycle(c_2)\in s_2$ the corresponding nesting graphs should be equivalent, and hence $C_f(c_1)=C_f(c_2)$.
    \item Next, consider the set of vertices~$I=(\midset(c_1)\cup \midset(c_2))\setminus\midset(t)$, that are on the noose of both children and not on the noose of $t$.  %
    We require for every $v\in I$ with missing neighbors sets $X_{c_1}(v)\in\noose(c_1)$ and $X_{c_2}(v)\in\noose(c_2)$, that $X_{c_1}(v)\cap X_{c_2}(v)=\emptyset$.%
    \item Lastly, for all faces $f$ current in both $s_1$ and $s_2$, given the mappings $M_1\in s_1$ and $M_2\in s_2$ such that $M_1(C_f)=(\ps^1,\pe^1)$ and $M_2(C_f)=(\ps^2,\pe^2)$, a clockwise traversal of the cycle of $C_f$ from $\ps^1$ encounters first $\pe^1$ then $\ps^2$ and finally $\pe^2$ (the order is not strict, $\ps^1=\pe^1=\ps^2=\pe^2$ or any equality is allowed). If one or both signatures have no cycle vertices specified by the mapping (empty set), no traversal occurs.
\end{enumerate}

Any signature pair $s_1,s_2$ that fails any of the above tests is not consistent and no parent signature is computed. Otherwise, to compute the combination, \algB\ does the following.
The set $\cycle(t)$ corresponds to $\cycle(c_1)\cup\cycle(c_2)\setminus C$ where $C$ corresponds to the nesting graphs of faces current in $\eta(c_1)$ and $\eta(c_2)$ and processed in $\eta(t)$.
For every face $f$ current in $\eta(c_1)$ and $\eta(c_2)$ processed in $\eta(t)$, $S_f$ is the set of
vertices in $C$ and $\inner(t)=S_f\cup \inner(c_1)\cup\inner(c_2)$.
The new mapping $M_t$ is a combination of $M_1$ and $M_2$ on the nesting graphs for the faces that are current  in $\eta(t)$. For such a current face $f$ there are four candidate vertices: $\ps^1,\pe^1,\ps^2,\pe^2$. The two vertices that are not on the noose of $\eta(t)$ are removed. If $\ps^1,\pe^2$ remain, $M_t(C_f)=(\ps^1,\pe^2)$, otherwise, $\pe^1,\ps^2$ remain and $M_t(C_f)=(\ps^2,\pe^1)$.
Lastly, $\noose(t)$ is the union of $\noose(c_1)$ and $\noose(c_2)$ without the sets $X_{t}(p)$ where $p\in\midset(c_1)\cap\midset(c_2)$, thus we get $\noose(t) = \langle X_t(p) \rangle_{p \in \midset(t)}$. Observe that the computed tuple indeed adheres to the definition of a signature.%

For every node $t$ with children $c_1, c_2$, we compute the combination $\algB(s_1, s_2)$ of all consistent signatures of children $s_1\in D(c_1), s_2\in D(c_2)$. We prove that exactly those signatures are the valid signatures of $t$.

\begin{lemma}\label{lma:algoB}
For an internal node $t\in T$ with children $c_1, c_2$, the table of its valid signatures can be computed using the following recurrence relation:
\begin{equation*}
  D(t)=\{\algB(s_1, s_2) \mid s_1\in D(c_1), s_2\in D(c_2), s_1,s_2 \text{ consistent}\}
\end{equation*}
\end{lemma}

\begin{proof}
Given $c_1,c_2,t \in V(T)$ such that $c_1,c_2$ are the children of $t$, to prove our recurrence relation, we show $\{\algB(s_1, s_2) \mid s_1\in D(c_1), s_2\in D(c_2), s_1,s_2 \text{ consistent}\}\subseteq D(t)$.
For $s_1\in D(c_1),s_2\in D(c_2)$ where $s_1=(\inner(c_1), \cycle(c_1), M_{c_1}, \noose(c_1))$ and $s_2=(\inner(c_2), \cycle(c_2), M_{c_2}, \noose(c_2))$ and $s_1$ and $s_2$ are consistent we show that $s_t = \algB(s_1, s_2)=(\inner(t), \cycle(t), M_t, \noose(t))$ is valid meaning there exists a corresponding partial solution $\sigma'=(\ssplitp, (N_{\dot{v}})_{\dot{v} \in \ssplitp}, \Gamma')$.

Since $s_1,s_2$ are valid a few properties are immediately verified: $\ssplitp$, the vertices in $\inner(t)$ and inside the cycles of $\cycle(t)$ are a union of the split vertices in each child's signature (where some nesting graph vertices are moved to $\inner(t)$), $\ssplitp\subseteq\ssplit$.
Additionally, for every $\dot{v}\in\ssplitp$, $N_{\dot{v}}\subseteq N_G(\orig(\dot{v}))$, since all coverage is inherited, and for every $v\in \orig(\ssplitp),\{N_{\dot{v}}\mid \dot{v}\in\copies(v)\}$ partitions $N_G(v)$. Additionally, while \algB\ modifies the nesting graph sets, it does not change the graphs themselves (it removes a subset of the children graphs for the parent set), the edge set induced by \ssplitp\ still corresponds to the guess made by Branching Rule~\ref{br:edge}.

For each pistil $p$ inside of $\eta(t)$, we want to show that they have no missing neighbors remaining. Since $s_1$ and $s_2$ are valid, all the pistils that were inside $\eta(c_1),\eta(c_2)$ are already covered, we must show that their remaining noose vertices that are inner vertices in $\eta(t)$ do not have remaining missing neighbors. This is done by \algB\ in the third step, 
for every $v\in I$ (as defined in test 3) where $X_{c_1}(v)\in\noose(c_1),X_{c_2}(v)\in\noose(c_2)$ we have $X_{c_1}(v)\subseteq N_G(v)$ %
and $X_{c_2}(v)\subseteq N_G(v)$ (by definition). Since $X_{c_1}(v)\cap X_{c_2}(v)=\emptyset$, by test 3, we know that $(N_G(v) \setminus X_{c_1}(v)) \cup (N_G(v) \setminus X_{c_2}(v)) = N_G(v)$.
Thus, the missing neighbors sets for non-noose vertices can be safely removed in $s_t$ as they are empty in the union of the two partial solutions. 
The remaining sets involve the vertices of $\midset(t)$, which are noose vertices which do not have to be covered in $s_t$.

We finally show that we can find a planar drawing $\Gamma'$ that extends $\Gamma$. Such drawings exist for $\eta(c_1)$ and $\eta(c_2)$, we call them $\Gamma_1$ and $\Gamma_2$ respectively, and we show that drawing $\Gamma'$ exists for the subgraph inside $\eta(t)$. The solution drawing for the faces not shared between $\eta(c_1)$ and $\eta(c_2)$ can directly be extended to $\Gamma'$ and is planar: \algB\ verifies in the first test that the only copies shared by both signatures' vertices were necessarily in nesting graphs that correspond to faces shared by both nooses, which ensures that there are no edges spanning from a face of $\eta(c_1)$ to a face of $\eta(c_2)$ which could create crossings. %
Next we explain how, for the remaining current faces present in both children nooses, we use the drawings of $\Gamma_1$ and $\Gamma_2$ to create $\Gamma'$. In $\Gamma_1$ and $\Gamma_2$, some copies cover pistils in only one subgraph. For those we can replicate the embedding in $\Gamma'$ as it was in both partial drawings. For a vertex $\dot{u}$ that has neighbors $u_1 \in G'_1$ and $u_2\in G'_2$, let us assume it does not have a planar embedding, meaning we find a crossing involving a split copy $\dot{x}$. This means that when traversing $f$, we can find a subsequence of vertices $u_1,x_1,u_2,x_2$ such that $x_1,x_2\in N_{\dot{x}}$. This can be extended to the vertices of the nesting graph: there must be two neighbors of $\dot{u}$ and two neighbors of $\dot{x}$ on the cycle that are alternating. This means that the two neighborhoods are not sequential in the cyclic ordering of the vertices on the cycle of the nesting graph. This contradicts the validity of the children signatures, and this situation is not possible, showing that we can create $\Gamma'$ by replicating the neighborhoods and embeddings in $\Gamma_1$ and $\Gamma_2$.
The mappings for nesting graphs corresponding to the faces now processed in $\eta(t)$, that were current in $\eta(c_1)$ and $\eta(c_2)$, are removed as we now can consider those faces processed and thus lie completely inside noose. The mappings for nesting graphs corresponding to faces still current in $\eta(t)$ that were current in $\eta(c_1)$ and $\eta(c_2)$ are found in the following manner. The cycle sections that were unused (between $\pe^1$ and $\ps^2$ and $\pe^2$ and $\ps^1$) are either removed as they are inside $\eta(t)$ and not able to cover pistils outside of it, or they become the new vertices the mapping designates.
Thus $\sigma'$ is a partial solution for $s_t$ and hence $s_t$ is valid.%

We now show $D(t)\subseteq\{\algB(s_1, s_2) \mid s_1\in D(c_1), s_2\in D(c_2), s_1,s_2 \text{ consistent}\}$.
Given a valid signature $s_t\in D(t)$, we show that it is possible to find two valid consistent signatures $s_1=(\inner(c_1), \cycle(c_1), M_{c_1}, \noose(c_1))$ for $c_1$ and $s_2=(\inner(c_2), \cycle(c_2), M_{c_2}, \noose(c_2))$ for $c_2$.
Since $s_t$ is valid there is a corresponding partial solution $\sigma'=(\ssplitp, (N_{\dot{v}})_{\dot{v} \in \ssplitp}, \Gamma')$ where $\Gamma'$ is the drawing of the solution graph inside $\eta(t)$. We restrict $\Gamma'$ to the graph~$G'_1$ inside $\eta(c_1)$, which is necessarily a partial solution for $G'_1$. Similarly we obtain a partial solution for $G'_2$. By definition of the signature of a partial solution, there exists a valid signature $s_1$ and $s_2$ for the two children nodes. We now show that these are consistent. Since $\Gamma'$ is planar, any vertex embedded in a face $f$ inside $\eta(c_1)$ can only be incident to vertices in $f$. Meaning the first test of \algB\ is passed. The shared current faces of $s_1$ and $s_2$ will also necessarily have the same nesting graph as they are the same face in $G_t'$: they induce the same nesting graphs, and thus the second test is passed. As for the third test, since $\sigma'$ is a partial solution, all of the inner vertices in $G_t'$ have no missing neighbors, meaning that the noose vertices of $\eta(c_1)$ and $\eta(c_2)$ that are not noose vertices of $\eta(t)$ are covered by a vertex in a face of $G'_1$ or $G'_2$, which will be reflected in the signature of the partial solution and no missing neighbor can remain, they are partitioned between the two signatures.
Lastly, we assume that the mappings created by \algB\ do not pass the fourth test, meaning w.l.o.g. in a face $f$'s nesting graph $C_f$ we encounter $\ps^1$ then $\ps^2$ then $\pe^1$ and $\pe^2$. This ordering is strict since the test is not passed. By construction, vertex $\ps^2$ stands for either, a pistil in $G'_1$ and $G'_2$, or a path of pistils and non pistils in $G'_1$ and $G'_2$ that were contracted to make $\ps^2$ in the nesting graph. In the first case this means that the pistil $\ps^2$ stands for is a noose pistil and $\ps^2=\pe^1$ and in the second case one of the vertices on the path is on the noose, and similarly since the path is contracted, then $\ps^2=\pe^1$. Since once the traversal of $f$ leaves $G'_1$ by passing the noose vertex, the next encountered pistils are pistils of $G'_2$ and $\pe^1$ cannot be any of the cycle vertex that stands for a pistil of $G'_2$ which contradicts our assumption.

Thus $s_1$ and $s_2$ are consistent, the recurrence relation is correct, and Algorithm~\algB\ computes all valid signatures.
\end{proof}

\begin{lemma}\label{lma:internalruntime}
  The table for one node is filled in $\mathcal{O}(N_s(k)^2k^4)$ time.
\end{lemma}

\begin{proof}
Algorithm $\algB$ first does four different tests.
To verify that the inner vertex sets of size at most $2k$ are disjoint requires quadratic time, we additionally need to compute the vertex sets of the nesting graph which can be done in linear time. Overall the first test requires $\bigoh(k^2)$ time.
Checking that the at most $20k$ nesting graphs are equal can be done in linear time for the second test, for a $\bigoh(k)$ runtime.
There are at most $20k$ noose vertices which have at most $k$ missing neighbors. Verifying that the missing neighbors sets are disjoint requires $\bigoh(k^4)$ time.
Lastly, a clockwise traversal of the cycle of length at most $4k$ requires linear time.
Overall these check are done for each signature pair in the table of valid signatures for the children of size at most $N_s(k)$, the number of possible signatures.
Overall, \algB\ can be fill the table for an internal node in $\bigoh(N_s(k)^2\cdot k^4)$ time.
\end{proof}

\subsection{Checking Signatures at the Root Node}\label{subsec:dp-root}

Finally, we check for the existence of a solution $\sigma$ by looking at the valid signatures of the children of the root node. In the initialization we split a single edge~$e_r = (r_1,r_2)$ into two edges by connecting the end points of $e_r$ to a new root node $r$. As a result $mid(r_1)=mid(r_2)=mid(e_r)$, and the union of the graphs inside $\eta(r_1)$ and $\eta(r_2)$ would make up the whole input graph.
With the previous section we have shown how to compute the tables $D(r_1)$ and $D(r_2)$. We now show how given $s_1\in D(r_1),s_2\in D(r_2)$ we decide if a solution exists.
We reuse here algorithm \algB\ to decide whether the two signatures are consistent. As there are no current faces anymore, the only relevant element in the tuple is the set $\inner$ of vertices that were used to cover the pistils in $G'$. We run one last test on the set $\ssplit\setminus\inner$. The edge set guessed by Branching Rule~\ref{br:edge} holds some of the coverage of the split vertices, but if a split vertex is not in $\inner$ and has a neighbor in $G_{\ssplit}$, then we must be able to realize those last edges. Hence we verify that the graph $G_{\ssplit}[\ssplit\setminus\inner]$ is planar, and if so, embed it inside any face of a partial solution. Thus, if this test is passed, we find that a solution $\sigma$ exists.

\begin{lemma}\label{lma:root}
For an instance $I$ of $\pREDP$ with input graph $G'$ and modified sphere-cut decomposition tree $T$ rooted on $r$ where $r_1,r_2$ are the children of $r$ with $G'_1,G'_2$ the subgraphs inside $\eta(r_1),\eta(r_2)$ respectively, such that $G'_1\cup G'_2=G'$, $I$ has a solution if and only if there exists valid consistent signatures for $r_1,r_2$.
\end{lemma}

\begin{proof}
Given that two valid consistent signatures $s_1,s_2$ exist, for $r_1,r_2\in V(T)$ with $G'_1,G'_2$ the subgraphs inside $\eta(r_1),\eta(r_2)$ respectively, such that $G'_1\cup G'_2=G'$ we show that $I$ has a solution.

From \cref{lma:algoB}, the existence of valid consistent signatures $s_1,s_2$ ensures that $D(r)\ne\emptyset$. 
For $s_r\in D(r)$ s.t. $s_r=(\inner(r), \cycle(r), M_{r}, \noose(r))$, $s_r$ is valid, meaning it corresponds to a partial solution $\sigma'=(\ssplitp, (N_{\dot{v}})_{\dot{v} \in \ssplitp}, \Gamma')$.
We can construct a solution for $I$ that looks as follows: 
$(\ssplit, \orig, \copies, (N_{\dot{v}})_{\dot{v} \in \ssplit}, \Gamma^*)$. 
The mappings $\copies$ and $\orig$ are those guessed in Branching Rule~\ref{br:vertex}.
We know that $\ssplitp = \inner(r) \subseteq \ssplit$. The missing vertices $\ssplit\setminus\ssplitp$ can be added in any face of the partial solution's drawing $\Gamma'$ to get the drawing $\Gamma^*$. We will argue later that $\Gamma^*$ is still planar.%

We now show for $(N_{\dot{v}})_{\dot{v} \in \ssplit}$, that the family $\{N_{\dot{v}} \mid \dot{v} \in \copiessol(v) \}$ is a partition of $N_G(v)$.
First consider a pistil~$p\in\eta(r_1)$, since $\midset(r_1)=\midset(r_2)$, $p\in\eta(r_2)$. When creating the set $\noose(r)$, \algB\ removes the sets $X_{t}(p)$ from $\noose(r)$ where $p\in\midset(c_1)\cap\midset(c_2)$, hence $\noose(r)=\emptyset$.
Now we make a case distinction on the different types of pistils. For a pistil $p\in G'$, w.l.o.g. assume $p\in V(G'_1)$ and $p\not\in \midset(r_1)$, since $s_1$ is valid, $p$ is covered. If $p\in\midset(r_1)$ then $p\in\midset(r_2)$ and since $\noose(r)=\emptyset$, $p$ is covered. For a split vertex $\dot{p}\in \inner(s_r)$, since $s_r$ is valid, $\dot{p}$ is covered. Lastly, if $\dot{p}\in \ssplit\setminus\inner(s_r)$, $\dot{p}$ has all of its missing neighbors in $\ssplit$: $\dot{p}$ cannot have neighbors in $\inner(s_r)$, as vertices in $\inner(s_r)$ are already completely covered. We only need to embed $G_{\ssplit}[\ssplit\setminus\inner(s_r)]$ in $\Gamma'$ to obtain the coverage of $\dot{p}$.

Since $s_r$ is valid, $\Gamma'$ is planar. The only vertices of $\ssplit$ not yet embedded in $\Gamma'$ are the vertices not in $\ssplit\setminus\inner$. The last test ran after \algB\ verifies that the graph $G_{\ssplit}[\ssplit\setminus\inner(s_r)]$ is planar, and since it is not a connected component of $G'$, it can be embedded into any of its faces and preserve planarity. Hence, there exists a planar drawing~$\Gamma^*$ of the solution.

Given a solution $I=(\ssplit, \orig, \copies, (N_{\dot{v}})_{\dot{v} \in \ssplit}, \Gamma^*)$ we now show that we can find two valid consistent signatures signatures $s_1,s_2$ for $r_1,r_2\in V(T)$ with $G'_1,G'_2$ the subgraphs inside $\eta(r_1),\eta(r_2)$ respectively, such that $G'_1\cup G'_2=G'$.
By definition, the drawing $\Gamma^*$ is the drawing of a graph extending $G'$, we remove from it any component not connected to $G'$. The vertices we remove are necessarily vertices in $\ssplit$ as all other vertices are in $G'$. We call $\ssplitp$ the remaining split vertices and $\Gamma'$ the drawing obtained. 
Since every pistil in $\Gamma^*$ was covered and we did not remove vertices connected to $G'$, every pistil in $\Gamma'$ is covered, and it is a planar drawing.
With these elements we construct $\sigma'=(\ssplitp, (N_{\dot{v}})_{\dot{v} \in \ssplitp}, \Gamma')$ a partial solution for $G'$. We can now obtain the signature of a partial solution (this is a special case with no current faces), $s_r=(\inner(r), \cycle(r), M_{r}, \noose(r))$ where $\ssplit=\inner(r)$. \cref{lma:algoB} tells us that given a valid signature for a parent node, there must be two valid consistent signatures for its children, hence, as $r$ corresponds to $G'$, its children must partition $G'$ in two and thus we find two valid consistent signatures signatures $s_1,s_2$ for $r_1,r_2\in V(T)$ with $G'_1,G'_2$ the subgraphs inside $\eta(r_1),\eta(r_2)$ respectively, such that $G'_1\cup G'_2=G'$.
\end{proof}

\begin{lemma}\label{lma:rtruntime}
Given two valid signature tables $D(r_1)$ and $D(r_2)$ where $\midset(r_1)=\midset(r_2)$ such that the union of the subgraphs inside $\eta(r_1),\eta(r_2)$ is the whole input graph, we can decide in $\mathcal{O}(N_s(k)^2k^4)$ time if our input is a yes instance.%
\end{lemma}

\begin{proof}
\cref{lma:internalruntime} showed that \algB\ requires $\mathcal{O}(N_s(k)^2k^4)$ time to compute the table for an internal node. On the root level the only difference is the final test that checks the planarity of the graph induced by the at most $2k$ vertices of $\ssplit$ unused in the signature which can be done in $\bigoh(k)$ time.
Thus, we can check for the existence of a solution in $\mathcal{O}(N_s(k)^2k^4)$ time.%
\end{proof}

This now gives us the necessary tools to prove \cref{thm:ssre-fpt}.
}

\appendixproofwithrestatable{\fptpssre*}
{

\begin{proof}
From the input instance of $\pSSRE$, we first find mappings between candidate vertices and copies, and determine how the graph induced by candidate vertices maps to copies. \cref{lma:brvertex,lma:bredge} show that an exhaustive search can do this in $\bigoh(2^k k^4)$ time. Next, we show with \cref{cor:reduc:outerp,outerplma,lma:bridgisouter} that we can apply transformations to the input graph in linear time to obtain a graph, whose sphere-cut decomposition we can compute. Thus we obtain an instance of $\pREDP$, that we solve using our dynamic programming algorithm.

The dynamic programming algorithm computes partial solutions starting from the leaves of the sphere-cut decomposition. By \cref{lma:algoA} we know that all valid signatures for all $m$~leaves can be correctly computed, and this can be done in $\mathcal{O}(N_s(k)k2^{2k}2k!)$ time per leaf (see \cref{lma:leafruntime}). The number of possible signatures $N_s(k)$ can be bounded using \cref{lma:sigcount}, to get $N_s(k)= 2^{2k} \cdot 2^{20k^{2}} \cdot 2^{O(k \log k)} \cdot \binom{8k}{2} = 2^{O(k^{2})}$. Those signatures can be enumerated in $\bigoh(N_s(k))$.
The dynamic programming now correctly computes all valid signatures for all $m-1$~internal nodes of the sphere-cut decomposition in bottom up fashion (see \cref{lma:algoB}), in $\mathcal{O}(N_s(k)^2k^4)$ time per internal node, as shown in \cref{lma:internalruntime}. Finally, we arrive at the root, where we can correctly determine whether a solution to \pSSRE\ exists, by \cref{lma:root}. This takes an additional $\mathcal{O}(N_s(k)^2k^4)$ time (see \cref{lma:rtruntime}).
Thus, we can solve any instance in $2^{O(k^2)} \cdot n^{O(1)}$ time
\end{proof}
}

\section{Conclusions}
\looseness=-1
We have introduced the embedded splitting number problem. %
However, fixed-parameter tractability is only established for the \pSSRE\ subproblem. The main open problem is to investigate the parameterized complexity of \pESN. A trivial \XP-algorithm for \pESN\ can provide appropriate inputs to \pSSRE\ as follows: check for any subset of up to $k$ vertices whether removing those vertices results in a planar input drawing, and branch on all such subsets.

\looseness=-1
Many variations of embedded splitting number are interesting for future work. For example, rather than aiming for planarity, we can utilize vertex splitting for crossing minimization. Other possible extensions can adapt the splitting operation, for example, %
the split operation allows both creating an additional copy of a vertex and re-embedding it, and the cost of these two parts can differ: simply re-embedding a vertex can be a cheaper operation. %

\subsection*{Acknowledgments}

We would like to thank an anonymous reviewer for their input to simplify the proof of~\cref{thm:hardnessone}.

\bibliographystyle{splncs04}
\bibliography{references}

\ifArxiv
\newpage
\appendix
\appendixProofText
\fi

\end{document}